\documentclass[letterpaper]{article} 
\usepackage{aaai22}  

\usepackage{times}  
\usepackage{helvet}  
\usepackage{courier}  
\usepackage[hyphens]{url}  
\usepackage{graphicx} 
\usepackage{natbib}  
\usepackage{caption} 
\DeclareCaptionStyle{ruled}{labelfont=normalfont,labelsep=colon,strut=off} 
\usepackage{algorithm}
\usepackage{algorithmic}

\usepackage{newfloat}
\usepackage{listings}

\usepackage{multirow}
\usepackage{tabulary}

\usepackage[utf8]{inputenc} 
\usepackage[T1]{fontenc}    
\usepackage{booktabs}       
\usepackage{amsfonts}       
\usepackage{nicefrac}       
\usepackage{microtype}      
\usepackage{amsthm}
\usepackage{amssymb}
\usepackage{xcolor}
\usepackage{bm}
\usepackage{mathrsfs}
\usepackage{algorithmic}
\usepackage{siunitx}
\usepackage{algorithm, float}
\usepackage{graphicx}
\usepackage{xcolor}
\usepackage{comment}
\usepackage{nccmath}
\usepackage{subfig}
\usepackage{amsmath}
\usepackage{bigints}
\usepackage{url}
\usepackage{natbib}
\allowdisplaybreaks

\newcommand\E{\mathbb E}
\newtheorem{defn}{Definition}
\newtheorem{theorem}{Theorem}
\newtheorem{theorem2}{Theorem}
\newtheorem{proposition}{Proposition}
\newtheorem{assumption}{Assumption}
\newtheorem{lemm}{Lemma}
\usepackage{xcolor}

\floatstyle{ruled}
\newfloat{listing}{tb}{lst}{}
\floatname{listing}{Listing}
\title{Decentralized Mean Field Games}
\author {
    Sriram Ganapathi Subramanian\textsuperscript{\rm 1, \rm 3},
    Matthew E. Taylor\textsuperscript{\rm 2, \rm 4},
    Mark Crowley\textsuperscript{\rm 1},
    Pascal Poupart\textsuperscript{\rm 1, \rm 3}
}
\affiliations {
    \textsuperscript{\rm 1} University of Waterloo, Waterloo, Ontario, Canada\\
    \textsuperscript{\rm 2} University of Alberta, Edmonton, Alberta, Canada\\
    \textsuperscript{\rm 3} Vector Institute, Toronto, Ontario, Canada\\
    \textsuperscript{\rm 4} Alberta Machine Intelligence Institute (Amii), Edmonton, Alberta, Canada\\
    s2ganapa@uwaterloo.ca, matthew.e.taylor@ualberta.ca, mcrowley@uwaterloo.ca, ppoupart@uwaterloo.ca
}

\begin{document}

\maketitle

\begin{abstract}
Multiagent reinforcement learning algorithms have not been widely adopted in large scale environments with many agents as they often scale poorly with the number of agents. Using mean field theory to aggregate agents has been proposed as a solution to this problem. However, almost all previous methods in this area make a strong assumption of a centralized system where all the agents in the environment learn the same policy and are effectively indistinguishable from each other. In this paper, we relax this assumption about indistinguishable agents and propose a new mean field system known as \emph{Decentralized Mean Field Games}, where each agent can be quite different from others. All agents learn independent policies in a decentralized fashion, based on their local observations. We define a theoretical solution concept for this system and provide a fixed point guarantee for a $Q$-learning based algorithm in this system. A practical consequence of our approach is that we can address a `chicken-and-egg' problem in empirical mean field reinforcement learning algorithms. Further, we provide $Q$-learning and actor-critic algorithms that use the decentralized mean field learning approach and give stronger performances compared to common baselines in this area. In our setting, agents do not need to be clones of each other and learn in a fully decentralized fashion. Hence, for the first time, we show the application of mean field learning methods in fully competitive environments, large-scale continuous action space environments, and other environments with heterogeneous agents. Importantly, we also apply the mean field method in a ride-sharing problem using a real-world dataset. We propose a decentralized solution to this problem, which is more practical than existing centralized training methods.
\end{abstract}

\section{Introduction}
Most multiagent reinforcement learning (MARL) algorithms are not tractable when applied to environments with many agents (infinite in the limit) as these algorithms are exponential in the number of agents \cite{busoniu2006multi}. One exception is a class of algorithms that use the mean field theory \citep{stanley1971phase} to approximate the many agent setting to a two agent setting, where the second agent is a mean field distribution of all agents representing the average effect of the population. This makes MARL algorithms tractable since, effectively, only two agents are being modelled. 
\citet{lasry2007mean} introduced the framework of a \emph{mean field game} (MFG), which incorporates the mean field theory in MARL. 
In MARL, the mean field can be a population state distribution \citep{huang2010large} or action  distribution \citep{pmlr-v80-yang18d}  of all the other agents in the environment. 

MFGs have three common assumptions. First, each agent does not have access to the local information of the other agents. However, it has access to accurate global information regarding the mean field of the population. Second, all agents in the environment are independent, homogeneous, and indistinguishable. Third, all agents maintain interactions with others only through the mean field. These assumptions (especially the first two) severely restrict the potential of using mean field methods in real-world environments. The first assumption is impractical, while the second assumption implies that all agents share the same state space, action space, reward function, and have the same objectives. 
Further, given these assumptions, prior works use centralized learning methods, where all agents learn and update a shared centralized policy. These two assumptions are only applicable to cooperative environments with extremely similar agents. Theoretically, the agent indices are omitted since all agents are interchangeable \citep{lasry2007mean}.
We will relax the first two assumptions. In our case, the agents are not interchangeable and can each formulate their own policies during learning that differs from others. Also, we will not assume the availability of the immediate global mean field. Instead, agents only have local information and use modelling techniques (similar to opponent modelling common in MARL \citep{hernandez2019survey}) to effectively model the mean field during the training process. 
We retain the assumption that each agent's impact on the environment is infinitesimal \cite{huang2010large}, and hence agents calculate best responses only to the mean field. Formulating best responses to each individual agent is intractable and unnecessary \citep{lasry2007mean}.


The solution concepts proposed by previous mean field methods have been either the centralized Nash equilibrium \citep{pmlr-v80-yang18d} or a closely related mean field equilibrium \citep{lasry2007mean}. These solution concepts are centralized as they require knowledge of the current policy of all other agents or other global information, even in non-cooperative environments. Verifying their existence is infeasible in many practical environments \citep{neumann1928theorie}. 

This work presents a new kind of mean field system, \emph{Decentralized Mean Field Games} (DMFGs), which uses a decentralized information structure. This new formulation of the mean field system relaxes the assumption of agents' indistinguishability and makes mean field methods applicable to numerous real-world settings. Subsequently, we provide a decentralized solution concept for learning in DMFGs, which will be more practical than the centralized solution concepts considered previously. We also provide a fixed point guarantee for a $Q$-learning based algorithm in this system. 

A `chicken-and-egg' problem exists in empirical mean field reinforcement learning algorithms where the mean field requires agent policies, yet the policies cannot be learned without the global mean field \citep{Yang2020overview}. We show that our formulation can address this problem, and we provide practical algorithms to learn in DMFGs. We test our algorithms in different types of many agent environments. We also provide an example of a ride-sharing application that simulates demand and supply based on a real-world dataset.


\section{Background}\label{sec:background}

\textbf{Stochastic Game}: An $N$-player stochastic (Markovian) game can be represented as a tuple $\langle \mathcal{S}, \mathcal{A}^1, \ldots, \mathcal{A}^N, r^1, \ldots, r^N, p, \gamma \rangle$, where $\mathcal{S}$ is the state space, $\mathcal{A}^j$ represents the action space of the agent $j \in \{1, \ldots, N\}$, and $r^j: \mathcal{S} \times \mathcal{A}^1 \times \cdots \times \mathcal{A}^N \xrightarrow{} \mathcal{R}$ represents the reward function of $j$. Also, $p: \mathcal{S} \times \mathcal{A}^1 \times \cdots \times \mathcal{A}^N \xrightarrow{} \Omega(\mathcal{S})$ represents the transition probability that determines the next state given the current state and the joint action of all agents. Here $\Omega(\mathcal{S})$ is a probability distribution over the state space. In the stochastic game, agents aim to maximize the expected discounted sum of rewards, with discount factor $\gamma \in [0,1)$.

Each agent $j$ in the stochastic game formulates a policy, which is denoted by $\pi^j: \mathcal{S} \xrightarrow{} \Omega(\mathcal{A}^j) $ where the joint policy is represented as $\boldsymbol{\pi} = [\pi^1, \ldots, \pi^N]$ for all $s \in \mathcal{S}$. Given an initial state $s$, the value function of the agent $j$ under the joint policy $\boldsymbol{\pi}$ can be represented as $v^j(s | \boldsymbol{\pi}) = \sum_{t=0}^\infty \gamma^t \E [r^j_t | s_0 = s, \boldsymbol{\pi}]$. The solution concept is the \emph{Nash equilibrium} (NE) \citep{hu2003nash}. This is  represented by a joint policy $\boldsymbol{\pi}_* = [\pi^1_*, \ldots, \pi^N_*]$, such that, for all $\pi^j$,  $v^j(s|\pi^j_*, \boldsymbol{\pi}^{-j}_*) \geq v^j(s|\pi^j, \boldsymbol{\pi}^{-j}_*)$, for all agents $j \in \{1, \ldots, N\}$ and all states $s \in \mathcal{S}$. The notation $\boldsymbol{\pi}^{-j}_*$ represents the joint policy of all agents except the agent $j$.

\textbf{Mean Field Game}: MFG was introduced as a framework to solve the stochastic game when the number of agents $N$ is very large \citep{lasry2007mean, huang2006large}. In this setting, calculating the best response to each individual opponent is intractable, so each agent responds to the aggregated state distribution  $\mu_t (s) \triangleq \lim_{N \xrightarrow{} \infty} \frac{\sum_{j=1}^N \boldsymbol{1} (s^j_t = s)}{N}$, known as the mean field. Here $\boldsymbol{1}$ denotes the indicator function, i.e. $\boldsymbol{1}(x) = 1$, if $x$ is true and $0$ otherwise. Let $\boldsymbol{\mu} \triangleq \{ \mu_t \}^\infty_{t=0}$. MFG assumes that all agents are identical (homogeneous), indistinguishable, and interchangeable \citep{lasry2007mean}. Given this assumption, the environment changes to a single-agent stochastic control problem where all agents share the same policy \citep{saldi2018discrete}. Hence, $\pi^1 = \cdots = \pi^N = \boldsymbol{\pi}$. The theoretical formulation focuses on a representative player, and the solution of this player (optimal policy) obtains the solution for the entire system. The value function of a representative agent can be given as $V(s, \boldsymbol{\pi}, \boldsymbol{\mu}) \triangleq \E_{\boldsymbol{\pi}} \big[ \sum_{t=0}^\infty \gamma^t r(s_t, a_t, \mu_t) | s_0 = s \big]
$, where $s$ and $a$ are the state and action of the representative agent, respectively. The transition dynamics is represented as $P(s_t, a_t, \mu_t)$, where the dependence is on the mean field distribution. All agents share the same reward function $r(s_t, a_t, \mu_t)$. The action comes from the policy $\pi_t(a_t | s_t, \mu_t)$.

The solution concept for the MFG is the mean field equilibrium (MFE). A tuple $(\boldsymbol{\pi}^*_{MFG}, \boldsymbol{\mu}^*_{MFG})$ is a mean field equilibrium if, for any policy $\boldsymbol{\pi}$, an initial state $s \in \mathcal{S}$ and given mean field $\boldsymbol{\mu}^*_{MFG}$, $
    V(s, \boldsymbol{\pi}^*_{MFG}, \boldsymbol{\mu}^*_{MFG}) \geq V(s, \boldsymbol{\pi}, \boldsymbol{\mu}^*_{MFG})$. 
Additionally, $\boldsymbol{\mu}^*_{MFG}$ is the mean field obtained when all agents play the same $\boldsymbol{\pi}^*_{MFG}$ at each $s \in \mathcal{S}$.


\textbf{Mean Field Reinforcement Learning} (MFRL): MFRL, introduced by \citet{pmlr-v80-yang18d}, is another approach for learning in stochastic games with a large number of agents. Here, the empirical mean action is used as the mean field, which each agent then uses to update its $Q$-function and Boltzmann policy. Each agent is assumed to have access to the global state, and it takes local action from this state. 
In MFRL, the $Q$-function of the agent $j$ is updated as, 
\begin{equation}  \label{eq:MFQ}
\begin{array}{l}
Q^j(s_t,a^j_t, \mu^a_{t})  \\   \quad \quad  = (1-\alpha) Q^j(s_t,a^j_t, \mu^a_{t}) + \alpha[r^j_t + \gamma v^j(s_{t+1})]
\end{array}
\end{equation}
\begin{align}
&\nonumber \mbox{where } \\ &
v^{j}(s_{t+1}) = \sum_{a^j_{t+1}}\pi^j(a^j_{t+1}|s_{t+1},\mu^a_{t})  Q^j(s_{t+1},a^j_{t+1}, \mu^a_{t}) \label{eq:yangvalueupdate} \\
& 
\mu^a_t = \frac{1}{\mathcal{N}} \sum_{j} a^j_t, a^j_t \sim \pi^j(\cdot|s_t,\mu^a_{t-1}) \label{eq:updatemeana} \\
& \pi^j(a^j_t|s_t, \mu^a_{t-1}) = \frac{\exp(-\hat{\beta} Q^j(s_t,a^j_t, \mu^a_{t-1})}{\sum_{a^{j'}_t\in A^j}\exp(-\hat{\beta} Q^j(s_t,a^{j'}_t,\mu^a_{t-1}))} \label{eq:updatepolicy}
\end{align}

\noindent where $s_t$ is the global old state, $s_{t+1}$ is  the global resulting state, $r^j_t$ is the reward of $j$ at time $t$, $v^j$ is the value function of $j$, $\mathcal{N}$ is the total number of agents, and $\hat{\beta}$ represents the Boltzmann parameter. The action $a^j$ is assumed to be discrete and represented using one-hot encoding. Like stochastic games, MFRL uses the NE as the solution concept. The global mean field $\mu^a_{t}$ captures the action distribution of all agents. To address this global limitation, \citet{pmlr-v80-yang18d} specify mean field calculation over certain neighbourhoods for each agent. However, an update using Eq.~\ref{eq:updatemeana} requires each agent to have access to all other agent policies or the global mean field to use the centralized concept (NE). 
We omit an expectation in Eq.~\ref{eq:yangvalueupdate} since \citet{pmlr-v80-yang18d} guaranteed that their updates will be greedy in the limit with infinite exploration (GLIE). 

From Eq.~\ref{eq:updatemeana} and Eq.~\ref{eq:updatepolicy}, it can be seen that the current action depends on the mean field and the mean field depends on the current action. To resolve this `chicken-and-egg' problem, \citet{pmlr-v80-yang18d} simply use the previous mean field action to decide the current action, as in Eq.~\ref{eq:updatepolicy}. This can lead to a loss of performance since the agents are formulating best responses to the previous mean field action $\mu^a_{t-1}$, while they are expected to respond to the current mean field action $\mu^a_{t}$. 

\section{Related Work}

MFGs were first proposed in \citet{huang2003individual}, while a comprehensive development of the system and principled application methods were given later in \citet{lasry2007mean}. Subsequently, learning algorithms were proposed for this framework. \citet{subramanian2019reinforcement} introduce a restrictive form of MFG (known as stationary MFG) and provide a model-free policy-gradient \citep{sutton1999policy, konda1999actor} algorithm along with convergence guarantees to a local Nash equilibrium. On similar lines, \citet{guo2019learning} provide a model-free $Q$-learning algorithm \citep{watkins1992q} for solving MFGs, also in the stationary setting. The assumptions in these works are difficult to verify in real-world environments. Particularly, \citet{guo2019learning} assume the presence of a game engine (simulator) that accurately provides mean field information to all agents at each time step, which is not practical in many environments. 
Further, other works depend on fictitious play updates for the mean field parameters \citep{hadikhanloo2019finite, elie2020convergence}, 
which involves the strong assumption that opponents play stationary strategies. All these papers use the centralized setting for the theory
and the experiments.

Prior works \citep{gomes2010discrete, adlakha2015equilibria, saldi2018discrete} have established the existence of a (centralized) mean field equilibrium in the discrete-time MFG under a discounted cost criterion, in finite and infinite-horizon settings. Authors have also studied the behaviour of iterative algorithms and provided theoretical analysis for learning of the non-stationary (centralized) mean field equilibrium in infinite-horizon settings \citep{wikecek2015stationary, wikecek2020discrete, anahtarci2019value}. We provide similar guarantees in the decentralized setting with possibly heterogeneous agents.

\citet{pmlr-v80-yang18d} introduces MFRL that uses a mean field approximation through the empirical mean action, and provides two practical algorithms that show good performance in large MARL settings. The approach is model-free and the algorithms do not need strong assumptions regarding the nature of the environment. However, they assume access to global information that needs a centralized setting for both theory and experiments. 
MFRL has been extended to multiple types \citep{Srirammtmfrl2020} and partially observable environments \citep{sriram2021partially}. However, unlike us, \citet{Srirammtmfrl2020} assumes that agents can be divided into a finite set of types, where agents within a type are homogeneous. 
\citet{sriram2021partially} relaxes the assumption of global information in MFRL, however, it still uses the Nash equilibrium as the solution concept.
Further, that work contains some assumptions regarding the existence of conjugate priors, which is hard to verify in real-world environments. Additional related work is provided in Appendix~\ref{sec:morerealtedwork}.

\section{Decentralized Mean Field Game}

The DMFG model is specified by $\langle \boldsymbol{\mathcal{S}}, \boldsymbol{\mathcal{A}}, p, \boldsymbol{R}, \mu_0 \rangle$, where $\boldsymbol{\mathcal{S}} = \mathcal{S}^1 \times \cdots \times \mathcal{S}^N$ represents the state space and $\boldsymbol{\mathcal{A}} =  \mathcal{A}^1 \times \cdots \times \mathcal{A}^N$ represents the joint action space. Here, $\mathcal{S}^j$ represents the state space of an agent $j \in \{ 1, \ldots, N \}$ and $\mathcal{A}^j$ represents the action space of $j$. As in MFGs, we are considering the infinite population limit of the game, where the set of agents $N$, satisfy $N \xrightarrow{} \infty$. 
Similar to the MFG formulation in several prior works \citep{lasry2007mean, huang2006large, saldi2018discrete}, we will specify that both the state and action spaces are Polish spaces. Particularly, $\mathcal{A}^j$ for all agents $j$, is a compact subset of a finite dimensional Euclidean space $\Re^d$ with the Euclidean distance norm $||\cdot||$. Since all agents share the same environment, for simplicity, we will also assume that the state spaces of all the agents are the same $\mathcal{S}^1 = \cdots = \mathcal{S}^N = \mathcal{S}$ and are locally compact. 
Since the state space is a complete separable metric space (Polish space), it is endowed with a metric $d_X$. The transition function $p: \mathcal{S} \times \mathcal{A}^j \times \mathcal{P}(\mathcal{S)} \xrightarrow{} \mathcal{P}(\mathcal{S})$ determines the next state of any $j$ given the current state and action of $j$, and the probability distribution of the state in the system (represented by the mean field). The reward function is represented as a set $\boldsymbol{R} = \{ R^1, \ldots, R^N \}$, where, $R^j: \mathcal{S} \times \mathcal{A}^j \times \mathcal{P}(\mathcal{S}) \xrightarrow{} [0, \infty)$ is the reward function of $j$.

 Recall that a DMFG has two major differences as compared to MFG and MFRL. 1) DMFG does not assume that the agents are indistinguishable and homogeneous (agent indices are retained). 2) DMFG does not assume that each agent can access the global mean field of the system. However, each agents' impact on the environment is infinitesimal, and therefore all agents formulate best responses only to the mean field of the system (no per-agent modelling is required).



As specified, in DMFG, the transition and reward functions for each agent depends on the mean field of the environment, represented by $\boldsymbol{\mu} \triangleq (\mu_t)_{t \geq 0}$, with the initial mean field represented as $\mu_0$. For DMFG the mean field can either correspond to the state distribution $\mu_t (s) \triangleq \lim_{N \xrightarrow{} \infty} \frac{\sum_{j=1}^N \boldsymbol{1} (s^j_t = s)}{N}$, or the action distribution $\mu^a_t$ as in Eq.~\ref{eq:updatemeana} (for discrete settings, or a mixture of Dirac measures for continuous spaces as used in \citet{anahtarci2019value}). Without a loss of generality, we use the mean field as the state distribution $\mu_t$ (represented by $\mathcal{P}(\mathcal{S})$), as done in prior works \citep{lasry2007mean, elliott2013discrete}. However, our setting and theoretical results will hold for the mean field as action distribution $\mu^a_t$ as well. 

In the DMFG, each agent $j$ will not have access to the true mean field of the system and instead use appropriate techniques to actively model the mean field through exploration. The agent $j$ holds an estimate of the actual mean field represented by $\mu^j_t$. Let $\boldsymbol{\mu}^j \triangleq (\mu^j_t)_{t \geq 0}$. Let, $\mathcal{M}$ be used to denote the set of mean fields $\{ \boldsymbol{\mu^j} \in \mathcal{P}(\mathcal{S})\}$. A Markov policy for an agent $j$ is a stochastic kernel on the action space $\mathcal{A}^j$ given the immediate local state ($s^j_t$) and the agent's current estimated mean field $\mu^j_t$, i.e. $\pi_t^j: \mathcal{S} \times \mathcal{P}(\mathcal{S}) \xrightarrow{} \mathcal{P}(\mathcal{A}^j)$, for each $t \geq 0$. Alternatively, a non-Markov policy will depend on the entire state-action history of game play. 
We will use $\Pi^j$ to denote a set of all policies (both Markov and non-Markov) for the agent $j$. 
Let $s^j_t$ represent the state of an agent $j$ at time $t$ and $a^j_t$ represent the action of $j$ at $t$. Then an agent $j$ tries to maximize the objective function given by the following equation (where $r^j$ denotes the immediate reward obtained by the agent $j$ and $\beta \in [0, 1)$ denotes the discount factor), 

\begin{equation}\label{eq:valuefunction}
    \begin{array}{l}
         J^j_{\boldsymbol{\mu}}(\pi^j) \triangleq \E^{\pi^j} [ \sum_{t=0}^{\infty} \beta^t r^j(s^j_t, a^j_t, \mu_t)].
    \end{array}
\end{equation}



In line with prior works \citep{pmlr-v80-yang18d, Srirammtmfrl2020}, we assume that each agent's sphere of influence is restricted by its neighbourhood, where it conducts exploration. Using this assumption, we assert that, after a finite $t$, the mean field estimate will accurately reflect the true mean field for $j$, in its neighbourhood denoted as $\mathcal{N}^j$. This assumption specifies that agents have full information in their neighbourhood, and they can use modelling techniques to obtain accurate mean field information within the neighbourhood (also refer to flocking from \citet{sarah2021mean}).

\begin{assumption}\label{assumption:limitofbelief}
There exists a finite time $T$ and a neighbourhood $\mathcal{N}^j$, such that for all $t > T$, the mean field estimate of an agent $j \in {1, \ldots, N}$ satisfies ($\forall s^j \in \mathcal{N}^j$) $\boldsymbol{\mu}^j(s^j) = \boldsymbol{\mu}(s^j)$. Also, $\forall s^j \in \mathcal{N}^j$, we have, $p(\cdot|s^j_t, a^j_t, \mu^j_t) = p(\cdot| s^j_t, a^j, \mu_t)$ and $r^j(\cdot|s^j_t, a^j_t, \mu^j_t) = r^j(\cdot|s^j_t, a^j_t, \mu_t)$.
\end{assumption}   

Let us define a set $\Phi: \mathcal{M} \xrightarrow{} 2^\Pi$ as $\Phi(\boldsymbol{\mu^j}) = \{ \pi^j \in \Pi^j: \pi^j \textrm{ is optimal for } \boldsymbol{\mu^j} \}$. Conversely, for $j$, we define a mapping $\Psi: \boldsymbol{\Pi} \xrightarrow{} \mathcal{M}$ as, given a policy $\pi^j \in \Pi^j$, the mean field state estimate $\boldsymbol{\mu}^j \triangleq \Psi(\pi^j)$ can be constructed as, 

\begin{equation}\label{eq:meanfielddetermination}
    \begin{array}{l}
         \mu^j_{t+1}(\cdot) = \bigintss_{\mathcal{S} \times \mathcal{A}^j}  p(\cdot|s^j_t, a^j_t, \mu_t)  \mathcal{P}^{\pi^j}(a^j_t| s^j_t, \mu^j_t) \mu^j_t(s^j_t).
   \end{array}
\end{equation}

\noindent Here $\mathcal{P}^{\pi^j}$ is a probability measure induced by $\pi^j$. 
Later (in Theorem~\ref{theorem:markovpolicies}) we will prove that restricting ourselves to Markov policies is sufficient in a DMFG, and hence $\mathcal{P}^{\pi^j} = \pi^j$. 




Now, we can define the \emph{decentralized mean field equilibrium} (DMFE) which is the solution concept for this game.

\begin{defn}\label{def:DMFE}
The decentralized mean field equilibrium of an agent $j$ is represented as a pair $(\pi^j_*$, $\mu^j_*) \in \Pi^j \times \mathcal{M}$ if $\pi^j_* \in \Phi(\mu^j_*)$ and $\mu^j_* = \Psi(\pi^j_*)$. Here $\pi^j_*$ is the best response to $\mu^j_*$ and $\mu^j_*$ is the mean field estimate of $j$ when it plays $\pi^j_*$.
\end{defn}

The important distinction between DMFE and centralized concepts, such as NE and MFE, is that DMFE does not rely on the policy information of other agents. MFE requires all agents to play the same policy, and NE requires all agents to have access to other agents' policies. DMFE has no such constraints. 
Hence, this decentralized solution concept is more practical
than NE and MFE.  
In Appendix~\ref{sec:differencesofmeanfieldsettings}, we summarize the major differences between the DMFG, MFG and MFRL. 

\section{Theoretical Results}\label{sec:theoriticalresults}

We provide a set of theorems that will first guarantee the existence of the DMFE in a DMFG. Further, we will show that a simple $Q$-learning update will converge to a fixed point representing the DMFE. 
We will borrow relevant results from prior works in centralized MFGs in our theoretical guarantees. Particularly, we aim to adapt the results and proof techniques in works by \citet{saldi2018discrete}, \citet{lasry2007mean}, and \citet{anahtarci2019value} to the decentralized setting. 
The statements of all our theorems are given here, while the proofs are in Appendices~\ref{sec:markovpolicy} -- \ref{sec:convergencetoequilibrium}.  

Similar to an existing result from centralized MFG \citep{lasry2007mean}, in the DMFG, restricting policies to only Markov policies would not lead to any loss of optimality. We use $\Pi^j_M$ to denote the set of Markov policies for the agent $j$.

\begin{theorem}\label{theorem:markovpolicies}
For any mean field, $\boldsymbol{\mu} \in \mathcal{M}$, and an agent $j \in \{1, \ldots, N\}$, we have,
\begin{equation}\label{eq:markovequation}
    \sup_{\pi^j \in \Pi^j} J^j_{\boldsymbol{\mu}}(\pi^j) = \sup_{\pi^j \in \Pi^j_M} J^j_{\boldsymbol{\mu}} (\pi^j). 
\end{equation}
\end{theorem}


Next, we show the existence of a DMFE under a set of assumptions similar to those previously used in the centralized MFG \citep{lasry2007mean, saldi2018discrete}. The assumptions pertain to bounding the reward function and imposing restrictions on the nature of the mean field (formal statements in Appendix~\ref{sec:existence}). We do not need stronger assumptions than those previously considered for the MFGs.

\begin{theorem}\label{theorem:existence}
An agent $j \in \{1, \ldots, N\}$ in the DMFG admits a decentralized mean field equilibrium $(\pi^j_*, \mu^j_*) \in \Pi^j \times \mathcal{M}$. 
\end{theorem}



We use $\mathcal{C}$ to denote a set containing bounded functions in $\mathcal{S}$. Now, we define a decentralized mean field operator ($H$),

\begin{equation}\label{eq:DMFGoperatordefinition}
\begin{array}{l}
    H : \mathcal{C} \times \mathcal{P}(\mathcal{S}) \ni (Q^j, \boldsymbol{\mu}^j) \\ \\ \quad \quad  \quad \xrightarrow{} (H_1(Q^j, \boldsymbol{\mu}^j), H_2(Q^j, \boldsymbol{\mu}^j)) \in \mathcal{C} \times \mathcal{P}(\mathcal{S})
    \end{array}
\end{equation}

\noindent where 

\begin{equation}
\begin{array}{l}
H_1 (Q^j, \boldsymbol{\mu}^j) (s^j_t,a^j_t)  \triangleq r^j(s^j_t,a^j_t,\mu_t) \\ \\ \quad   + \beta \bigintss_{\mathcal{S}} Q^j_{\max_{a^j}} (s^j_{t+1}, a^j, \mu^j_{t+1}) p(s^j_{t+1} | s^j_t, a^j_t, \mu_t)
\end{array}
\end{equation}



\begin{equation}\label{eq:H2definition}
\begin{array}{l}
H_2 (Q^j, \boldsymbol{\mu}^j) (\cdot)   \\ \quad \quad  \quad  \quad  \triangleq \bigintss_{\mathcal{S} \times \mathcal{A}^j} p \Big(\cdot| s^j_t, \pi^j(s^j_t,Q^j_t,\mu^j_t), \mu_t \Big) \mu^j_t (s)
\end{array}
\end{equation}

\noindent for an agent $j$. Here, $\pi^j$ is a maximiser of the operator $H_1$. 

For the rest of the theoretical results, we consider a set of assumptions different from those needed for Theorem~\ref{theorem:existence}. Here we assume that the reward and transition functions are Lipschitz continuous with a suitable Lipschitz constant. The Lipschitz continuity assumption is quite common in the mean field literature \citep{pmlr-v80-yang18d, lasry2007mean, Srirammtmfrl2020}. We also consider some further assumptions regarding the nature of gradients of the value function (formal statements in Appendix~\ref{sec:proofofoperatorexistence}). All assumptions are similar to those considered before for the analysis in the centralized MFGs \citep{lasry2007mean, anahtarci2019value, huang2010large}. 
First, we provide a theorem regarding the nature of operator $H$. Then, we provide another theorem showing that $H$ is a contraction.

\begin{theorem}\label{theorem:operatormapping}
The decentralized mean field operator $H$ is well-defined, i.e., this operator maps $ \mathcal{C} \times \mathcal{P}(\mathcal{S})$ to itself. 
\end{theorem}



\begin{theorem}\label{theorem:contractionofH}
Let $\mathcal{B}$ represent the space of bounded functions in $\mathcal{S}$. Then the mapping $H: \mathcal{C} \times \mathcal{P} (\mathcal{S}) \xrightarrow{} \mathcal{C} \times \mathcal{P}(\mathcal{S})$ is a contraction in the norm of $\mathcal{B}(\mathcal{S})$. 
\end{theorem}


Since $H$ is a contraction, by the Banach fixed point theorem \citep{shukla2016generalized}, we can obtain that fixed point using the $Q$ iteration algorithm given by Alg.~\ref{alg:Qldmfg}. We provide this result in the next theorem.

\begin{algorithm}[ht]
\caption{Q-learning for DMFG}
\label{alg:Qldmfg}
\begin{algorithmic}[1]
\STATE For each agent $j \in \{1, \ldots, N \}$, start with initial $Q$-function $Q^j_0$ and the initial mean field state estimate $\mu^j_0$
\WHILE{$(Q^j_n, \mu^j_n) \neq (Q^j_{n-1}, \mu^j_{n-1}) $}
\STATE $(Q^j_{n+1}, \mu^j_{n+1}) = H (Q^j_n, \mu^j_n)$
\ENDWHILE
\STATE Return the fixed point $(Q^j_*, \mu^j_*)$ of $H$
\end{algorithmic}
\end{algorithm}


\begin{theorem}\label{theorem:convergencetoequilibrium}
Let the $Q$-updates in Algorithm~\ref{alg:Qldmfg} converge to $(Q^j_*, \mu_*^j)$ for an agent $j \in \{1, \ldots, N\}$. Then, we can construct a policy $\pi^j_*$ from $Q^j_*$ using the relation, 
\begin{equation}
    \pi^j_*(s^j) = \arg \max_{a^j \in \mathcal{A}^j} Q^j_*(s^j, a^j, \mu^j_*).
\end{equation}
Then the pair $(\pi^j_*, \mu_*^j)$ is a DMFE. 
\end{theorem}


Hence, from the above results, we have proved that a DMFG admits a DMFE, and an iterative algorithm using $Q$-updates can arrive at this equilibrium. This provides a fixed point for Alg.~\ref{alg:Qldmfg}, to which the $Q$-values converge.






\section{Algorithms}\label{sec:algorithms}
We will apply the idea of decentralized updates to the model-free MFRL framework. We modify the update equations in MFRL (Section~\ref{sec:background}) and make them decentralized, where agents only observe their local state and do not have access to the immediate global mean field. 
Our new updates are: 

\begin{equation}
  \begin{array}{l}
Q^j(s^j_t, a^j_t, \mu^{j,a}_{t}) \\ \quad  = (1-\alpha) Q^j(s^j_t,a^j_t, \mu^{j,a}_{t}) + \alpha[r^j_t + \gamma v^j(s^j_{t+1})] \label{eq:DMFQ}
  \end{array}
\end{equation}
\begin{align}
&\nonumber \mbox{where } \\ & v^{j}(s^j_{t+1}) = \sum_{a^j_{t+1}}\pi^j(a^j_{t+1}|s^j_{t+1},\mu^{j,a}_{t+1})  Q^j(s^j_{t+1},a^j_{t+1}, \mu^{j,a}_{t+1}) \label{eq:Dvalueupdate} \\
& 
\mu^{j,a}_{t} = f^j(s^j_t, \hat{\mu}^{j,a}_{t-1}) \label{eq:Dupdatemeana} \\
& \mbox{and } \pi^j(a^j_t|s^j_t, \mu^{j,a}_{t}) = \frac{\exp(-\hat{\beta} Q^j(s^j_t,a^j_t, \mu^{j,a}_{t}))}{\sum_{a^{j'}_t\in A^j}\exp(-\hat{\beta} Q^j(s^j_t,a^{j'}_t,\mu^{j,a}_{t}))} \label{eq:Dupdatepolicy}
\end{align}

\noindent Here, $s^j_t$ is the local state and $\mu^{j,a}_{t}$ is the mean field action estimate for the agent $j$ at time $t$ and $\hat{\mu}^{j,a}_{t-1}$ is the observed local mean field action of $j$ at $t-1$. Other variables have the same meaning as Eq.~\ref{eq:MFQ} -- Eq.~\ref{eq:updatepolicy}. In Eq.~\ref{eq:Dupdatemeana}, the mean field estimate for $j$ is updated using a function of the current state and the previous local mean field. Opponent modelling techniques commonly used in MARL \citep{hernandez2019survey} can be used here. We use the technique of \citet{he2016opponent}, that used a neural network to model the opponent agent(s). In our case, we use a fully connected neural network (2 Relu layers of 50 nodes and an output softmax layer) to model the mean field action. The network takes the current state and previous mean field action as inputs and outputs the estimated current mean field. This network is trained using a mean square error between the estimated mean field action from the network (Eq.~\ref{eq:Dupdatemeana}) and the observed local mean field (local observation of actions of other agents $\hat{\mu}^{j,a}_{t}$) after action execution. The policy in Eq.~\ref{eq:Dupdatepolicy} does not suffer from the `chicken-and-egg' problem of Eq.~\ref{eq:updatepolicy} since it depends on the current mean field estimate, unlike MFRL which used the previous global mean field in Eq.~\ref{eq:updatepolicy}.

We provide a neural network-based $Q$-learning implementation for our update equations, namely \emph{Decentralized Mean Field Game $Q$-learning} (DMFG-QL), and an actor-critic implementation, \emph{Decentralized Mean Field Game Actor-Critic} (DMFG-AC). 
Detailed description of the algorithms are in Appendix~\ref{sec:algorithmimplementation} (see Algs.~\ref{alg:POGMFGQL} and \ref{alg:POGMFGAC}). A complexity analysis is in Appendix~\ref{sec:complexityanalysisappendix}, and hyperparameter details are in Appendix~\ref{sec:Hyperparameters}. 

\section{Experiments and Results}\label{sec:experiments}

In this section we study the performance of our algorithms. The code for experiments has been open-sourced \citep{sourcecode}. We provide the important elements of our domains here, while the complete details are in Appendix~\ref{sec:experimentaldetails}.

The first five domains belong to the MAgent environment \citep{zheng2018magent}. We run the experiments in two phases, training and execution. Analogous to experiments conducted in previous mean field studies \citep{pmlr-v80-yang18d, Srirammtmfrl2020}, all agents train against other agents playing the same algorithm for 2000 games. This is similar to multiagent training using self-play \citep{shoham2003multi}. The trained agents then enter into an execution phase, where the trained policies are simply executed.
The execution is run for 100 games, where algorithms may compete against each other. We consider three baselines, independent $Q$-learning  (IL) \citep{tan1993multi}, mean field $Q$-learning (MFQ) \citep{pmlr-v80-yang18d}, and mean field actor-critic (MFAC) \citep{pmlr-v80-yang18d}. Each agent in our implementations learns in a decentralized fashion, where it maintains its own networks and learns from local experiences. This is unlike centralized training in prior works \citep{pmlr-v80-yang18d, guo2019learning}. 
We repeat experiments 30 times, and report the mean and standard deviation.  Wall clock times are given in Appendix~\ref{sec:wallclocktimes}.




First, we consider the mixed cooperative-competitive Battle game \citep{zheng2018magent}. 
This domain consists of two teams of 25 agents each. Each agent is expected to cooperate with agents within the team and compete against agents of the other team to win the battle. 
 For the training phase, we plot the cumulative rewards per episode obtained by the agents of the first team for each algorithm in Fig.~\ref{fig:battleresults}(a). The performance of the second team is also similar (our environment is not zero-sum). From the results, we see that DMFG-QL performs best while the others fall into a local optimum and do not get the high rewards. The DMFG-AC algorithm comes second. It has been noted previously \citep{pmlr-v80-yang18d} that $Q$-learning algorithms often perform better compared to their actor-critic counterparts in mean field environments. 
 MFQ and MFAC (using the previous mean field information) performs poorly compared to DMFG-QL and DMFG-AC (using the current estimates). 
 Finally, IL loses out to others due to its independent nature.
 In execution, one team trained using one algorithm competes against another team from a different algorithm. We plot the percentage of games won by each algorithm in a competition against DMFG-QL and DMFG-AC. A game is won by the team that kills more of its opponents. 
The performances are in Fig.~\ref{fig:battleresults}(b) and (c), where DMFG-QL performs best. In Appendix~\ref{sec:meanfieldmodelling}, we show that DFMG-QL can accurately model the true mean field in the Battle game. 

\begin{figure}[ht!]
	\subfloat[ Training]{{\includegraphics[width=0.45\textwidth, height=3cm]{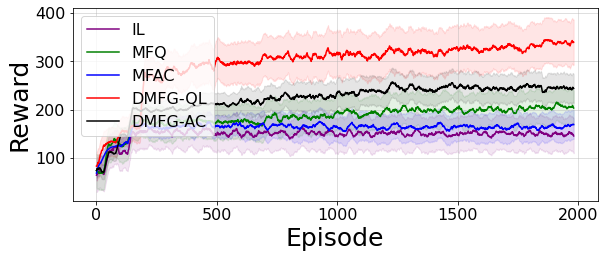} }}
	\\
	\subfloat[Execution vs. DMFG-QL]{{\includegraphics[width=0.23\textwidth, height=3cm]{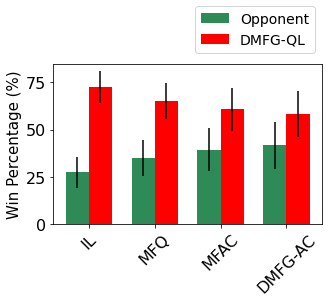} }}
	\subfloat[Execution vs. DMFG-AC]{{\includegraphics[width=0.22\textwidth, height=3cm]{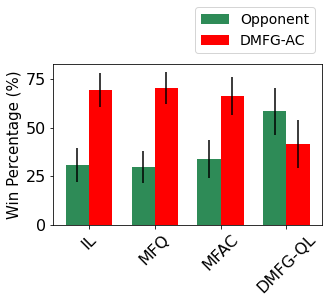} }}

  \caption{Battle results. In (a) solid lines show average and shaded regions represent standard deviation. In (b) and (c) bars are average and black lines represent standard deviation.}%
	\label{fig:battleresults}
\end{figure}

The second domain is the heterogeneous Combined Arms environment. This domain is a mixed setting similar to Battle except that each team consists of two different types of agents, ranged and melee, with distinct action spaces. Each team has 15 ranged and 10 melee agents. This environment is different from those considered in \citet{Srirammtmfrl2020}, which formulated each team as a distinct type, where agents within a team are homogeneous. The ranged agents are faster and attack further, but can be killed quickly. The melee agents are slower but are harder to kill. 
We leave out MFQ and MFAC for this experiment, since both these algorithms require the presence of fully homogeneous agents. The experimental procedure is the same as in Battle. From the results we see that DMFG-QL performs best
in both phases
(see Fig.~\ref{fig:combinedarmsresults}).

\begin{figure}[ht!]
	\subfloat[ Training]{{\includegraphics[width=0.31\textwidth, height=3cm]{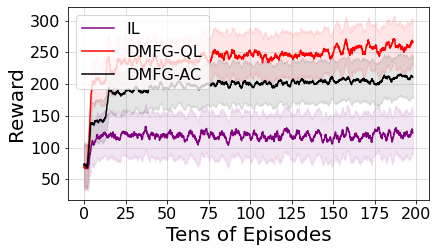} }}
		\subfloat[Execution ]{{\includegraphics[width=0.15\textwidth, height=3cm]{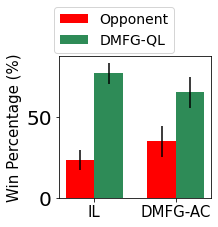} }}

  \caption{Combined Arms results}%
	\label{fig:combinedarmsresults}
\end{figure}

Next is the fully competitive Gather environment. This contains 30 agents trying to capture limited food. All the agents compete against each other for capturing food and could resort to killing others when the food becomes scarce. 
We plot the average rewards obtained by each of the five algorithms 
in the training phase (Fig.~\ref{fig:gatherresults}(a)). 
DMFG-QL once again obtains the maximum performance. In competitive environments, actively formulating the best responses to the current strategies of opponents is crucial for good performances. Predictably, the MFQ and MFAC algorithms (relying on previous information) lose out. 
For execution, we sample (at random) six agents from each of the five algorithms to make a total of 30. We plot the percentage of games won by each algorithm in a total of 100 games. A game is determined to have been won by the agent obtaining the most rewards.
Again, DMFG-QL shows the best performance during execution (Fig.~\ref{fig:gatherresults}(b)). 

\begin{figure}[t]
	\subfloat[ Training]{{\includegraphics[width=0.24\textwidth, height=3cm]{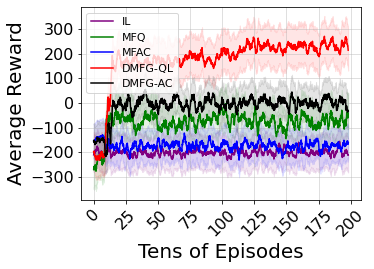} }}
		\subfloat[Execution ]{{\includegraphics[width=0.20\textwidth, height=3cm]{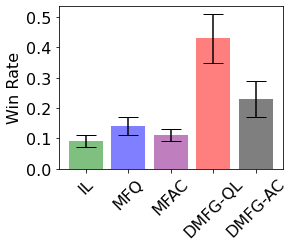} }}

  \caption{Gather results}%
	\label{fig:gatherresults}
\end{figure}


The next domain is a fully cooperative Tiger-Deer environment. In this environment, a team of tigers aims to kill deer. The deer are assumed to be part of the environment moving randomly, while the tigers are agents that learn to coordinate with each other to kill the deer. At least two tigers need to attack a deer in unison to gain large rewards. Our environment has 20 tigers and 101 deer. 
In the training phase, we plot the average reward obtained by the tigers (Fig.~\ref{fig:tigerresults}(a)).  The performance of MFQ almost matches that of DMFG-QL and the performance of DMFG-AC matches MFAC. In a cooperative environment, best responses to actively changing strategies of other agents are not as critical as in competitive environments. 
Here all agents aid each other and using the previous time information (as done in MFQ and MFAC) does not hurt performance as much. 
For execution, a set of 20 tigers from each algorithm execute their policy for 100 games. We plot the average number of deer killed by the tigers for each algorithm. 
DMFG-QL gives the best 
performance
(Fig.~\ref{fig:tigerresults}(b)). 

\begin{figure}[ht!]
	\subfloat[ Training]{{\includegraphics[width=0.22\textwidth, height=3cm]{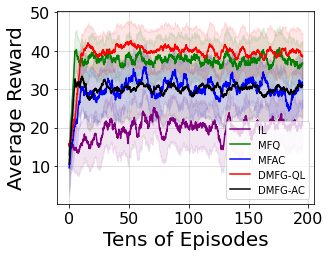} }}
		\subfloat[ Execution]{{\includegraphics[width=0.22\textwidth, height=3cm]{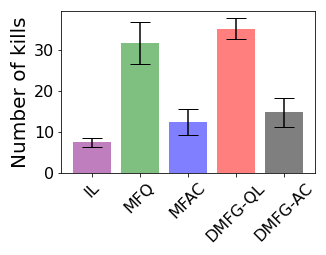} }}

  \caption{Tiger-Deer results}%
	\label{fig:tigerresults}
\end{figure}

The next domain is the continuous action Waterworld domain, first introduced by \citet{gupta2017cooperative}. This is also a fully cooperative domain similar to Tiger-Deer, where a group of 25 pursuer agents aim to capture a set of food in the environment while actively avoiding poison. The action space corresponds to a continuous thrust variable. 
We implement 
DMFG-AC in this domain, where the mean field is a mixture of Dirac deltas of actions taken by all agents. The experimental procedure is the same as Tiger-Deer. 
For this continuous action space environment, we use proximal policy optimization (PPO) \citep{schulman2017proximal} and deep deterministic policy gradient (DDPG) \citep{lillicrap2015continuous}, as baselines. 
We see that DMFG-AC obtains the best performance in both phases (refer to  Fig.~\ref{fig:waterworldresults}(a) and (b)). 

\begin{figure}[t]
	\subfloat[  Training]{{\includegraphics[width=0.22\textwidth, height=3cm]{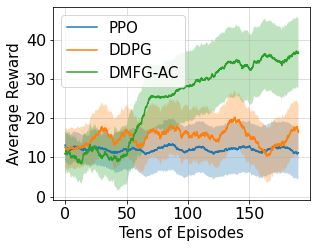} }}
		\subfloat[ Execution]{{\includegraphics[width=0.22\textwidth, height=3cm]{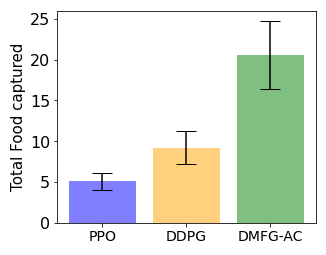} }}
  \caption{Waterworld results}%
	\label{fig:waterworldresults}
\end{figure}


Our final environment is a real-world \emph{Ride-pool Matching Problem} (RMP) introduced by \citet{javier2017ondemand}. This problem pertains to improving the efficiency of vehicles satisfying ride requests as part of ride-sharing platforms such as UberPool. In our environment, ride requests come from the open source New York Yellow Taxi dataset \citep{newyork2016}. The road network (represented as a grid with a finite set of nodes or road intersections) contains a simulated set of vehicles (agents) that aim to serve the user requests. Further details about this domain are in Appendix~\ref{sec:ridesharingdetails}. 
We consider two baselines in this environment. The first is the method from \citet{javier2017ondemand}, which used a constrained optimization (CO) approach to match ride requests to vehicles. This approach is hard to scale and is myopic in assigning requests (it does not consider future rewards). The second baseline is the Neural Approximate Dynamic Programming (NeurADP) method from \citet{shah2020neural}, which used a (centralized) DQN algorithm to learn a value function for effective mapping of requests. This approach assumes all agents are homogenous (i.e., having the same capacity and preferences), which is impractical. To keep comparisons fair, we consider a decentralized version of NeurADP as our baseline. 
Finally, we implement DMFG-QL for this problem where the mean-field corresponds to the distribution of ride requests at every node in the environment.  

Similar to prior approaches \citep{lowalekar2019zac, shah2020neural}, we use the service rate (total percentage of requests served) as the comparison metric. We train NeurADP and DMFG-QL using a set of eight consecutive days of training data and test all the performances in a previously unseen test set of six days. 
The test results are reported as an average (per day) of performances in the test set pertaining to three different hyperparameters. The first is the capacity of the vehicle ($c$) varied from 8 to 12, the second is the maximum allowed waiting time ($\tau$) varied from 520 seconds to 640 seconds, and the last is the number of vehicles ($N$), varied from 80 to 120. The results in Figs.~\ref{fig:ridesharing}(a-c) show that DMFG-QL outperforms the baselines in all our test cases. 
The mean field estimates in DMFG-QL help predict the distribution of ride requests in the environment, based on which agents can choose ride requests strategically. If agents choose orders that lead to destinations with a high percentage of requests, they will be able to serve more requests in the future. Thus, DMFG-QL outperforms the NeurADP method (which does not maintain a mean field). We note that the service rate for all algorithms is low in this study (10\% -- 30\%), since we are considering fewer vehicles compared to prior works \citep{shah2020neural} due to the computational requirements of being decentralized. In practice our training is completely parallelizable and this is not a limitation of our approach. Also, from Fig.~\ref{fig:ridesharing}(c), an increase in the number of vehicles, increases the service rate.
In Fig.~\ref{fig:ridesharing}(d) we plot the performance of the three algorithms for a single test day (24 hours --- midnight to midnight). During certain times of the day (e.g., 5 am), the ride demand is low, and all approaches satisfy a large proportion of requests. However, during the other times of the day, when the demand is high, the DMFG-QL satisfies more requests than the baselines, showing its relative superiority.


\begin{figure}
\subfloat[\# of Vehicles]{{\includegraphics[width=0.21\textwidth, height=3cm]{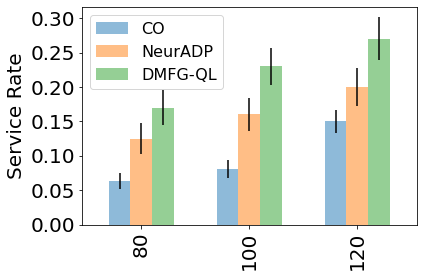} }} 
	\quad
	\subfloat[Maximum Pickup Delay]{{\includegraphics[width=0.21\textwidth, height=3cm]{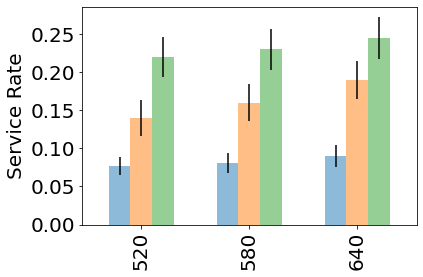} }}
	\\
	\subfloat[Capacity]{{\includegraphics[width=0.21\textwidth, height=3cm]{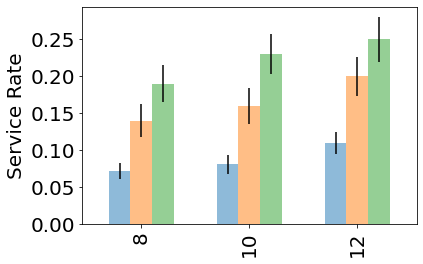} }}
	\quad
	\subfloat[Single Day Test]{{\includegraphics[width=0.21\textwidth, height=3cm]{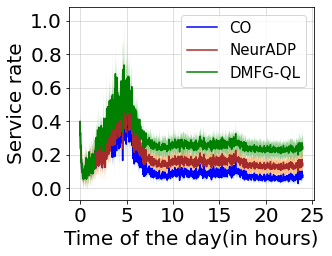} }}
  \caption{Results for the ride-sharing experiment. For (a), (b) and (c), we start with a prototypical configuration of $c$=10, $\tau$ = 580, and $N$ = 100, and then vary the different parameters. Figures (a), (b) and (c) share the same legend given in (a). 
  }%
	\label{fig:ridesharing}
\end{figure}

\section{Conclusion}

In this paper, 
we relaxed two strong assumptions in prior work on using mean field methods in RL. We introduced the DMFG framework, where agents are not assumed to have global information and are not homogeneous. All agents learn in a decentralized fashion, which contrasts with centralized procedures in prior work.
Theoretically, we proved that the DMFG will have a suitable solution concept, DMFE. Also, we proved that a $Q$-learning based algorithm will find the DMFE. Further, we provided a principled method to address the `chicken-and-egg' problem in MFRL, and demonstrated performances in a variety of environments (including RMP). 

For future work, we would like to extend our theoretical analysis to the function approximation setting and analyze the convergence of policy gradient algorithms. Empirically, we could consider other real-world applications like autonomous driving and problems on demand and supply optimization.

\section*{Acknowledgements}
Resources used in preparing this research at the University of Waterloo were provided by the province of Ontario and the government of Canada through CIFAR, NRC, NSERC and companies sponsoring the Vector Institute. 
Part of this work has taken place in the Intelligent Robot Learning (IRL) Lab at the University of Alberta, which is supported in part by research grants from the Alberta Machine Intelligence Institute (Amii); a Canada CIFAR AI Chair, CIFAR; Compute Canada; and NSERC.

\bibliography{aaai22.bib}

\appendix

\newpage

\section{Proof of Theorem~\ref{theorem:markovpolicies}}\label{sec:markovpolicy}

For the proof, we will follow the idea of using discounted occupancy measures, as discussed in \citet{putterman1994markov}. Let us consider an agent $j \in \{1, \ldots, N \}$. Let $\boldsymbol{\mu} \in \mathcal{M}$ and $\pi^j \in \Pi^j$ be arbitrary. Then from the Ionescu-Tulcea theorem (see Theorem 3.3 in \citet{lattimore2020bandit}), a mean field distribution $\boldsymbol{\mu}$ on $\mathcal{S}$ and a policy $\pi^j$ defines a unique probability measure $\mathcal{P}^{\pi^j}$ on $\mathcal{S} \times \mathcal{A}^j$. 

We will begin with a definition of a discounted occupancy measure for some mean field $\boldsymbol{\mu}$. Consider an arbitrary policy $\pi^j \in \Pi^j$. Then, the discounted occupancy measure $\nu^{\pi^j}_{\boldsymbol{\mu}} \in \mathcal{F}(\mathcal{S} \times \mathcal{A})$ induced by this policy on the set $\mathcal{S} \times \mathcal{A}^j$ with mean field $\boldsymbol{\mu}$ is given by the equation 

\begin{equation}\label{eq:occupancydefinition}
    \nu^{\pi^j}(s^j, a^j | \mu) = \sum_{t=0}^\infty \beta^t \mathcal{P}^{\pi^j} (s^j_t = s^j , a^j_t = a^j | \mu_t = \mu).
\end{equation}

The discounted occupancy measure assigns a measure to the given state-action pair, which is the infinite discounted sum of the process that has a mean field of $\boldsymbol{\mu}$ and follows the policy $\pi^j$, with the agent $j$ hitting at various times the state $s^j$ and taking the action $a^j$ at that state. 

The discounted measure has a property that the value function can be represented as a product of the immediate reward and the discounted occupancy measure. To prove this, consider (from Eq.~\ref{eq:valuefunction}), 

\begin{equation}\label{eq:relationwithoccupancymeasure}
    \begin{array}{l}
         J^j_{\boldsymbol{\mu}}(\pi^j) = \E^{\pi^j} [ \sum_{t=0}^{\infty} \beta^t r^j(s^j_t, a^j_t, \mu_t)]
         \\ \\
          = \bigintss_{(s^j \in \mathcal{S}, a^j \in \mathcal{A}^j)} \sum_{t=0}^{\infty} \beta^t \E^{\pi^j} [r^j(s^j_t, a^j_t, \mu_t)
         \\ 
         \quad \quad \quad \quad \quad \quad \quad \quad \quad \quad  \mathcal{I}(s^j = s^j_t , a^j = a^j_t, \mu_t = \mu)]
         \\ \\ 
          \overset{1}{=}  \bigintss_{(s^j \in \mathcal{S}, a^j \in \mathcal{A}^j)} \sum_{t=0}^{\infty} \beta^t \E^{\pi^j} [r^j(s^j, a^j, \mu)\\ \quad \quad \quad \quad \quad \quad \quad \quad  \quad \quad  \mathcal{I}((s^j_t = s^j , a^j_t = a^j, \mu_t = \mu)]
         \\ \\ 
         \overset{2}{=}  \bigintss_{(s^j \in \mathcal{S}, a^j \in \mathcal{A}^j)} r^j(s^j, a^j, \mu) \\ \quad \quad \quad \quad \quad  \sum_{t=0}^{\infty} \beta^t  \E^{\pi^j}  [\mathcal{I}(s^j_t = s^j , a^j_t = a^j, \mu_t = \mu)]
         \\ \\ 
         =  \bigintss_{(s^j \in \mathcal{S}, a^j \in \mathcal{A}^j)} r^j(s^j, a^j, \mu) \\ \quad \quad \quad \quad \quad \quad   \sum_{t=0}^{\infty} \beta^t \mathcal{P}^{\pi^j}(s^j_t = s^j , a^j_t = a^j | \mu_t = \mu)
         \\ \\ 
         =  \bigintss_{(s^j \in \mathcal{S}, a^j \in \mathcal{A}^j)} r^j(s^j, a^j, \mu)\nu^{\pi^j}(s^j, a^j | \mu)
         \\ \\
    \end{array}
\end{equation}

Here the notation $\mathcal{I}(x)$ is an indicator function, which equals 1 if the $x$ is true and 0 if $x$ is false. The sum of all indicators for all the state-action pairs at a given $t$ equals 1. The step (1) uses this property of the indicator function, to drop the time indices of the reward function.  The step (2) is from Lebesgue's dominated convergence theorem \citep{bartle2014elements}.  

Now we state a lemma needed for our proof. 

\begin{lemm}\label{lemm:occupancymeasure}
For an agent $j \in \{1, \ldots, N \}$ and policy $\pi^j$, given mean field $\boldsymbol{\mu}$, there exists a Markov policy $\hat{\pi}^j \in \Pi^j_M$ such that 
\begin{equation}
    \nu^{\hat{\pi}^j} (s^j, a^j | \mu) = \nu^{\pi^j}(s^j, a^j | \mu) 
\end{equation}
\end{lemm}

\begin{proof}

For a mean field $\boldsymbol{\mu}$, we define an occupancy measure over the state space $\Tilde{\nu}^{\pi^j}(s^j| \mu) \triangleq \bigintss_{a^j} \nu^{\pi^j} (s^j, a^j| \mu)$. 

Now consider a Markov policy $\hat{\pi}^j$ as follows, 
\begin{equation}\label{eq:definitionofmarkovpolicy}
\begin{array}{l}
    \hat{\pi}^j(a^j|s^j, \mu) = \frac{\nu^{\pi^j}(s^j, a^j |  \mu)}{\Tilde{\nu}^{\pi^j}(s^j | \mu)}, \quad \textrm{ if }  \Tilde{\nu}^{\pi^j}(s^j, \mu) \neq 0
    \\ \\ \textrm{and} \\ \\ 
    \hat{\pi}^j(a^j|s^j, \mu) = \pi^j_0(a^j), \quad \textrm{ if }  \Tilde{\nu}^{\pi^j}(s^j, \mu) = 0
    \end{array}
\end{equation}

Here, $\pi^j_0(a^j) \in \Pi^j_M$ is an arbitrary policy.

Consider, 

\begin{equation}\label{eq:expressionforstatemarkov}
    \begin{array}{l}
        \Tilde{\nu}^{\hat{\pi}^j}(s^j | \mu) = \bigintss_{a^j} \nu^{\hat{\pi}^j}(s^j, a^j| \mu)
        \\ \\ 
        = \bigintss_{a^j}\sum_{t=0}^\infty \beta^t \mathcal{P}^{\hat{\pi}^j} (s^j_t = s^j , a^j_t = a^j | \mu_t = \mu)
        \\ \\ 
        \overset{1}{=} \sum_{t=0}^\infty \beta^t \mathcal{P}^{\hat{\pi}^j} ( s^j_t = s^j | \mu_t = \mu)
        \\ \\ 
        = \nu_0(s^j) + \sum_{t=1}^\infty \beta^t  \mathcal{P}^{\hat{\pi}^j} ( s^j_t = s^j | \mu_t = \mu)
        \\ \\ 
        = \nu_0(s^j) + \beta \sum_{t=1}^\infty \beta^{t-1}  \bigintss_{s'^j \in \mathcal{S}} \bigintss_{a'^j \in \mathcal{A}^j}
        \\ \\ 
         \mathcal{P}^{\hat{\pi}^j}[s_{t-1}^j = s'^j, a_{t-1}^j = a'^j| \mu_{t-1} = \mu']p(s^j|s'^j,a'^j, \mu') 
        \\ \\ 
        = \nu_0(s^j) + \beta \sum_{t=1}^\infty \beta^{t-1}   \bigintss_{s'^j \in \mathcal{S}}  \\ \\  \quad \quad \quad \quad  \quad \quad \quad \quad   \mathcal{P}^{\hat{\pi}^j}[s_{t-1} = s'|  \mu_{t-1} = \mu'] \\ \\ \quad \quad   \quad \quad \quad \quad   \bigintss_{a'^j \in \mathcal{A}^j} \hat{\pi}^j(a'^j|s'^j, \mu') p(s^j|s'^j,a'^j, \mu') 
        \\ \\ 
        = \nu_0(s^j) \\ \\ + \beta \bigintss_{s'^j \in \mathcal{S}} \sum_{t=1}^\infty \beta^{t-1} \mathcal{P}^{\hat{\pi}^j}[s^j_{t-1} = s'^j | \mu_{t-1} = \mu']  \\ \\ 
         \bigintss_{a'^j \in \mathcal{A}^j} \hat{\pi}^j
         (a'^j|s'^j, \mu') p(s^j|s'^j,a'^j, \mu_{t-1} = \mu') 
       \\ \\ 
        \overset{2}{=} \nu_0(s^j) + \beta \bigintss_{s'^j \in \mathcal{S}}  \Tilde{\nu}^{\hat{\pi}^j}(s'^j | \mu') \\ \\ 
         \quad \quad \quad \quad \bigintss_{a'^j \in \mathcal{A}^j} \hat{\pi}^j(a'^j|s'^j, \mu') p(s^j|s'^j,a'^j, \mu') 
    \end{array}
\end{equation}

Here $\nu_0 (s^j)$ is the initial distribution of the state of the agent $j$. To obtain expression (2), see that this follows from expression (1).

Now consider, 

\begin{equation}\label{eq:expressionforstate}
    \begin{array}{l}
     \Tilde{\nu}^{\pi^j}(s^j| \mu) =   \sum_{t=0}^\infty \beta^t \mathcal{P}^{\pi^j} ( s^j_t = s^j | \mu_t = \mu)
        \\ \\ 
    = \nu_0(s^j) + \sum_{t=1}^\infty \beta^t  \mathcal{P}^{\pi^j} ( s^j_t = s^j | \mu_t = \mu)
        \\ \\ 
        = \nu_0(s^j) + \beta \sum_{t=1}^\infty \beta^{t-1}  \bigintss_{s'^j \in \mathcal{S}} \bigintss_{a'^j \in \mathcal{A}^j}
        \\ \\ 
         \mathcal{P}^{\pi^j}[s^j_{t-1} = s'^j, a^j_{t-1} = a'^j| \mu_{t-1} = \mu']p(s^j|s'^j,a'^j, \mu') 
                \\ \\ 
        = \nu_0(s^j) + \beta \bigintss_{s'^j \in \mathcal{S}} \bigintss_{a'^j \in \mathcal{A}^j}\sum_{t=1}^\infty \beta^{t-1}  
        \\ \\ 
         \mathcal{P}^{\pi^j}[s^j_{t-1} = s'^j, a^j_{t-1} = a'^j| \mu_{t-1} = \mu']p(s^j|s'^j,a'^j, \mu') 
                        \\ \\ 
        \overset{1}{=} \nu_0(s^j)  \\ \\  \quad  \quad   +  \beta \bigintss_{s'^j \in \mathcal{S}} \bigintss_{a'^j \in \mathcal{A}^j}
        \nu^{\pi^j}(s'^j, a'^j | \mu') p(s^j|s'^j,a'^j, \mu')
        \\ \\ 
        \overset{2}{=} \nu_0(s^j) + \beta \bigintss_{s'^j \in \mathcal{S}} \bigintss_{a'^j \in \mathcal{A}^j} \\ \\ 
        \quad  \quad  \quad  \quad \quad  \hat{\pi}^j(a'^j|s'^j, \mu') \Tilde{\nu}^{\pi^j}(s'^j | \mu')  p(s^j|s'^j,a'^j, \mu') 
      \\ \\ 
        = \nu_0(s^j)   + \beta \bigintss_{s'^j \in \mathcal{S}} \Tilde{\nu}^{\pi^j}(s'^j | \mu')  \\ \\ \quad  \quad  \quad  \quad \quad \quad  \quad  \quad \bigintss_{a'^j \in \mathcal{A}^j}
        \hat{\pi}^j(a'^j|s'^j, \mu')  p(s^j|s'^j,a'^j, \mu') 
        
    \end{array}
\end{equation}

The (1) is from Eq.~\ref{eq:occupancydefinition} and (2) is from Eq.~\ref{eq:definitionofmarkovpolicy}. From both Eq.~\ref{eq:expressionforstate} and Eq.~\ref{eq:expressionforstatemarkov}, we find that both the discounted state occupation frequency can be recursively expressed as a term depending on the state occupation frequency at the previous time step and a few other terms that are the same for both the Eq.~\ref{eq:expressionforstatemarkov} and Eq.~\ref{eq:expressionforstate}. Since the initial state distribution is the same for both policies, we can conclude that  $ \Tilde{\nu}^{\hat{\pi}^j}(s^j | \mu)  =  \Tilde{\nu}^{\pi^j}(s^j| \mu)  $. 

Now, consider 

\begin{equation}\label{eq:nuforstateactionmarkov}
    \begin{array}{l}
        \nu^{\hat{\pi}^j}(s^j, a^j | \mu) = \sum_{t=0}^\infty \beta^t \mathcal{P}^{\hat{\pi}^j}(s^j_t = s^j , a^j_t = a^j | \mu_t = \mu) 
        \\ \\ 
         =  \sum_{t=0}^\infty \beta^t \mathcal{P}^{\hat{\pi}^j} (a^j_t = a^j | s^j_t = s^j, \mu_t = \mu) \\ \\ \quad \quad \quad \quad \quad \quad \quad \quad \quad \quad \mathcal{P}^{\hat{\pi}^j} ( s^j_t = s^j | \mu_t = \mu)
         \\ \\ 
         = \sum_{t=0}^\infty \beta^t  \hat{\pi}^j (a^j_t = a^j | s^j_t = s^j, \mu_t = \mu) \\ \\ 
         \quad \quad \quad \quad \quad \quad \quad \quad \quad \quad \mathcal{P}^{\hat{\pi}^j} ( s^j_t = s^j | \mu_t = \mu)
    \end{array}
\end{equation}

Also, from Eq.~\ref{eq:definitionofmarkovpolicy}, we have 
\begin{equation}\label{eq:nuforstateactionnonmarkov}
    \begin{array}{l}
      \nu^{\pi^j}(s^j, a^j |  \mu) =  \hat{\pi}^j(a^j|s^j, \mu) \times  \Tilde{\nu}^{\pi^j}(s^j | \mu)  
      \\ \\ 
      = \hat{\pi}^j(a^j|s^j, \mu) \times  \Tilde{\nu}^{\hat{\pi}^j}(s^j | \mu)
      \\ \\ 
      =  \hat{\pi}^j(a^j|s^j, \mu)   \sum_{t=0}^\infty \beta^t \mathcal{P}^{\hat{\pi}^j} ( s^j_t = s^j | \mu_t = \mu)
       \\ \\ 
      =   \sum_{t=0}^\infty \beta^t \hat{\pi}^j(a^j|s^j, \mu)   \mathcal{P}^{\hat{\pi}^j} ( s^j_t = s^j | \mu_t = \mu)
    \end{array}
\end{equation}

From Eq.~\ref{eq:nuforstateactionmarkov} and Eq.~\ref{eq:nuforstateactionnonmarkov} we find that $\nu^{\hat{\pi}^j}(s^j, a^j | \mu) = \nu^{\pi^j}(s^j, a^j |  \mu) $, due to the property of the Markov policy $\hat{\pi}^j$.



       
\end{proof}

\begin{theorem2}
For any mean field, $\boldsymbol{\mu} \in \mathcal{M}$, and an agent $j \in \{1, \ldots, N\}$, we have,
\begin{equation*}
    \sup_{\pi^j \in \Pi^j} J^j_{\boldsymbol{\mu}}(\pi^j) = \sup_{\pi^j \in \Pi^j_M} J^j_{\boldsymbol{\mu}} (\pi^j).
\end{equation*}

\begin{proof}

Now, we aim to show that  $\sup_{\pi^j \in \Pi^j} J^j_{\boldsymbol{\mu}}(\pi^j) = \sup_{\pi^j \in \Pi^j_M} J^j_{\boldsymbol{\mu}} (\pi^j) $. Since $\Pi^j_M \subset \Pi^j$, we know that $  \sup_{\pi^j \in \Pi^j_M} J^j_{\boldsymbol{\mu}} (\pi^j)  \leq \sup_{\pi^j \in \Pi^j} J^j_{\boldsymbol{\mu}}(\pi^j) $. To show the equality we aim to prove  $ \sup_{\pi^j \in \Pi^j_M} J^j_{\boldsymbol{\mu}} (\pi^j)  \geq \sup_{\pi^j \in \Pi^j} J^j_{\boldsymbol{\mu}}(\pi^j) $ as well.

To show this result we use Lemma~\ref{lemm:occupancymeasure}. Consider a policy $\pi^j$ and a Markov policy $\hat{\pi}^j$, such that $\nu^{\hat{\pi}^j}_{\boldsymbol{\mu}} = \nu^{\pi^j}_{\boldsymbol{\mu}}$. Now using the Eq.~\ref{eq:relationwithoccupancymeasure}, we have 

\begin{equation}\label{eq:supermumrelation}
    \begin{array}{l}
       J^j_{\boldsymbol{\mu}}(\pi^j) = \bigintss_{(s^j \in \mathcal{S}, a^j \in \mathcal{A}^j)} r^j(s^j, a^j, \mu)\nu^{\pi^j}(s^j, a^j | \mu)
       \\ \\ 
        \overset{1}{=} \bigintss_{(s^j \in \mathcal{S}, a^j \in \mathcal{A}^j)} r^j(s^j, a^j, \mu)\nu^{\hat{\pi}^j}(s^j, a^j | \mu)
        \\ \\
        \leq \sup_{\hat{\pi}^j \in \Pi^j_M} \bigintss_{(s^j \in \mathcal{S}, a^j \in \mathcal{A}^j)} r^j(s^j_t, a^j_t, \mu_t)\nu^{\hat{\pi}^j}(s^j, a^j | \mu) 
        \\ \\ 
        = \sup_{\hat{\pi}^j \in \Pi^j_M} J^j_{\boldsymbol{\mu}}(\hat{\pi}^j)
    \end{array}
\end{equation}

Here (1) is from Lemma~\ref{lemm:occupancymeasure}. Now, from Eq.~\ref{eq:supermumrelation}, and taking supermum on both sides, we get that $\sup_{\pi^j \in \Pi^j} J^j_{\boldsymbol{\mu}}(\pi^j)  \leq \sup_{\pi^j \in \Pi^j_M} J^j_{\boldsymbol{\mu}} (\pi^j) $, which concludes our proof. 

\end{proof}
\newpage


\end{theorem2}

\section{Proof of Theorem~\ref{theorem:existence}}\label{sec:existence}

The proof of this theorem follows the Theorem~1 in \citet{saldi2018discrete}. We aim to extend it to the decentralized setting as mentioned in Section~\ref{sec:theoriticalresults}. 

Before starting the proof, we will state some assumptions, 

\begin{assumption}\label{assumption:rewardfunction}
The reward function for all agents $j \in \{1, \ldots, N\}$ is bounded and continuous. 
\end{assumption}


Let us define a function $w: \mathcal{S} \xrightarrow{} [1, \infty)$, such that there exists an increasing sequence of compact subsets $\{ K_n\}_{n \geq 1}$ of $\mathcal{S}$ where 
\begin{equation}
    \lim_{n \xrightarrow{} \infty} \inf_{s \in \mathcal{S} \backslash K_n} w(s) = \infty.
\end{equation}

That is, the value of $w(s)$ gets close to infinity in this limit. Here, the $w$ can be regarded as a continuous moment function \citep{saldi2018discrete}. 

\begin{assumption}\label{assumption:transitionbound}
There exists a positive $\alpha$ such that the following holds (for all agents $j \in \{1, \ldots, N\}$),
\begin{equation}
    \sup_{(a^j,\mu^j) \in \mathcal{A}^j \times P(\mathcal{S})} \bigintss_{\mathcal{S}} w(s^j_{t+1}) p(s^j_{t+1}|s^j_t,a^j_t,\mu_t) \leq \alpha w(s^j_t)
\end{equation}

Also, consider two different moment functions $w$ and $v$, then $\alpha$ satisfies, 

\begin{equation}
    \begin{array}{l}
\sup_{(a^j,\mu^j) \in \mathcal{A}^j \times P(\mathcal{S})} \Big |\bigintss_{\mathcal{S}} w(s^j_{t+1}) p(s^j_{t+1}|s^j_t,a^j_t,\mu_t) \\ \\ - \bigintss_{\mathcal{S}} v(s^j_{t+1}) p(s^j_{t+1}|s^j_t,a^j_t,\mu_t)) \Big| \leq \alpha | w(s^j_t) - v(s^j_t)|
\end{array}
\end{equation}

\end{assumption}

\begin{assumption}\label{assumption:initialmeanfield}
The initial mean field $\mu_0$ satisfies 
\begin{equation}
    \int_\mathcal{S} w(s^j) \mu_0(s^j
    ) \triangleq M < \infty
\end{equation}
\end{assumption}

\begin{assumption}\label{assumption:rewardbound}
There exist a $\gamma \geq 1$ and a positive scalar $\mathcal{R}$ such that for all $t \geq 0$, we define $M_t \triangleq \gamma^t \mathcal{R}$, then 
\begin{equation}
     \sup_{(a^j,\mu) \in \mathcal{A}^j \times P(\mathcal{S})} r^j(s^j_t,a^j_t,\mu_t) \leq M_t w(s^j_t)
\end{equation}

\end{assumption}

Now, we will state a few definitions required for our proof. For any function $g: \mathcal{S} \xrightarrow{} \Re$, we define the $w$-norm as: 
\begin{equation}\label{eq:vnormsimple}
    ||g||_w \triangleq \sup_{s \in \mathcal{S}} \frac{|g(s)|}{w(s)}
\end{equation}

Let $B_w(\mathcal{S})$ be the Banach space of all real valued measurable functions on $\mathcal{S}$ with a finite $w$-norm. 

Also, for any signed measure $\mu$ in the space of $\mathcal{S}$, the $w$-norm can be defined as 

\begin{equation}\label{eq:signedvnorm}
    ||\mu||_w \triangleq \sup_{g \in B_w(\mathcal{S}): ||g||_w \leq 1} \Big| \int_\mathcal{S} g(s) \mu (s)  \Big|
\end{equation}

Further, we define another set 

\begin{equation}
    \mathcal{T}_w(\mathcal{S}) \triangleq \{  \mu \in \mathcal{P}(\mathcal{S}): ||\mu||_w < \infty \}
\end{equation}

Also, for $t \geq 0$, we define 
\begin{equation}
    \mathcal{T}^t_w(\mathcal{S}) \triangleq \Big\{ \mu \in \mathcal{T}_w(\mathcal{S}): \int_\mathcal{S} w(s) \mu(s) \leq \alpha^t M   \Big\}
\end{equation}

\noindent where the constant $M$ is obtained from Assumption~\ref{assumption:initialmeanfield}. Like our notation $\mathcal{P}(\mathcal{S})$, we are using the notation $\mathcal{P}(\mathcal{S} \times \mathcal{A}^j)$ to denote the probability of a state-action pair. Let us consider an element $\nu \in \mathcal{P}(\mathcal{S} \times \mathcal{A}^j)$, and use the notation $\nu_1 \triangleq \nu(\cdot \times \mathcal{A}^j)$ to denote the state marginal of $\nu$. Now we define a new set, 

\begin{equation}\label{eq:thedefinitionofT}
    \mathcal{T}^t_w(\mathcal{S} \times \mathcal{A}^j) \triangleq \Big\{ \nu \in \mathcal{P}(\mathcal{S} \times \mathcal{A}^j): \nu_1 \in \mathcal{T}^t_w(\mathcal{S}) \Big\}. 
\end{equation}

For each $t \geq 0$, we will define 
\begin{equation}
    L_t \triangleq \sum_{k=t}^\infty (\beta \alpha)^{k-t} M_k
\end{equation}

Here the constants $\alpha$ and $M_k$ are obtained from Assumption~\ref{assumption:transitionbound} and Assumption~\ref{assumption:rewardbound} respectively. 

It follows that the following equation holds, 

\begin{equation}\label{eq:ltequation}
    L_t = M_t + (\beta \alpha) L_{t+1}
\end{equation}

Let $C_w(S)$ denote the Banach space of all real valued bounded measurable functions on $\mathcal{S}$ with a finite $w$-norm.

Let us define a set
\begin{equation} \label{eq:vnormbound}
    C^t_w(\mathcal{S}) \triangleq \{u \in C_w(\mathcal{S}): ||u||_w \leq L_{t}\}
\end{equation}

For an agent $j \in \{1, \ldots, N \}$, let us consider an operator 

\begin{equation}
\begin{array}{l}
  T^{\boldsymbol{\mu}}_t u^j(s_t) =  \max_{a_t^j \in \mathcal{A}^j} \Big [r(s_t^j,a_t^j,\mu_{t})  \\ \\ \quad \quad \quad \quad \quad \quad \quad  + \beta \bigintss_{\mathcal{S}} u^j(s^j_{t+1})p(s^j_{t+1}|s^j_t, a^j_t, \mu_{t})\Big ] 
  \end{array}
\end{equation}

\noindent where $u^j: \mathcal{S} \xrightarrow{} \Re$



We will first prove the following lemma. 

\begin{lemm}\label{lemma:operatorlemma}
For all $t \geq 0$ and a given mean field $\boldsymbol{\mu}$, the operator $T_t^{\boldsymbol{\mu}}$ maps $C^{t+1}_w(\mathcal{S})$ into $C^{t}_w(\mathcal{S}).$ Also, this operator will satisfy 
\begin{equation}\label{eq:contractionequation}
    || T^{\boldsymbol{\mu}}_t u - T^{\boldsymbol{\mu}}_t x||_w \leq \alpha \beta ||u - x||_w 
\end{equation}
for any $u,x \in C_w (\mathcal{S})$. 
\end{lemm}

\begin{proof}
Let $u \in C^{t+1}_w(\mathcal{S})$. We have the following relation, 

\begin{equation}
\begin{array}{l}
    ||T^{\boldsymbol{\mu}}_t u||_w
    \leq \sup_{(s^j_t,a^j_t) \in \mathcal{S} \times \mathcal{A}^j} \\ \\  \quad \quad \quad \quad \quad \frac{ \Big| r^j(s^j_t, a^j_t, \mu_t) + \beta \bigintss_S u(s^j_{t+1}) p(s^j_{t+1}|s^j_t,a^j_t,\mu_t) \Big|}{w(s^j_t)}
    \\ \\ 
    \leq \sup_{(s^j_t,a^j_t) \in \mathcal{S} \times \mathcal{A}^j} \frac{M_t w(s^j_t) + \beta \alpha L_{t+1} w(s^j_t)}{w(s^j_t)}
    \\ \\
    = M_t + \beta \alpha L_{t+1} 
    \\ \\
    = L_t
\end{array}    
\end{equation}

The second last step is from Assumption~\ref{assumption:rewardbound} and Assumption~\ref{assumption:transitionbound}. Also, we use the bound on the $w$-norm from Eq.~\ref{eq:vnormbound}. The last step is from Eq.~\ref{eq:ltequation}. This proves the first statement. 

To prove the second statement, without loss of generality, let us assume that $u_0 \geq x_0$. Now consider, 

\begin{equation}
\begin{array}{l}
    || T^{\boldsymbol{
    \mu}}_t u - T^{\boldsymbol{\mu}}_t x||_w   =
    \\ \\
    \sup_{(s^j_t) \in \mathcal{S}} 
    \frac{ \max_{a^j_t} \Big| r^j(s^j_t,a^j_t,\mu_t) + \beta \bigintsss_{\mathcal{S}} u(s^j_{t+1}) p(s^j_{t+1}| s^j_t,a^j_t,\mu_t)}{w(s^j_t)}
    \\ \\
    \quad \quad \quad \quad \quad \quad  \frac{- r^j(s^j_t,a^j_t, \mu_t) - \beta \bigintsss_{\mathcal{S}} x(s^j_{t+1}) p(s^j_{t+1}| s^j_t, a^j_t, \mu_t) \Big|}{w(s^j_t)}
    \\ \\
    = 
   \sup_{(s^j_t) \in \mathcal{S}}  \frac{\max_{a^j_t} \Big| \beta \Big(\bigintsss_{\mathcal{S}} u(s^j_{t+1}) p(s^j_{t+1}| s^j_t,a^j_t,\mu_t)}{w(s^j_t)}
    \\ \\
    \quad \quad \quad \quad \quad \quad \frac{- \bigintsss_{\mathcal{S}} x(s^j_{t+1}) p(s^j_{t+1}| s^j_t, a^j_t, \mu_t) \Big) \Big|}{w(s^j_t)}
    \\ \\
    = \sup_{(s^j_t) \in \mathcal{S}} \frac{\max_{a^j_t}  \Big|\beta\bigintsss_{\mathcal{S}}  p(s^j_{t+1}| s^j_t, a^j_t, \mu_t) (u(s^j_{t+1}) - x(s^j_{t+1})) \Big| }{w(s^j_t)}
    \\ \\ 
    \leq \sup_{(s^j_t) \in \mathcal{S}} \frac{ \alpha \beta  \Big |(u(s^j_{t}) - x(s^j_{t})) \Big| }{w(s^j_t)}
    \\ \\ 
    = \alpha \beta ||u - x||_w
 
 \end{array}   
\end{equation}

We apply Assumption~\ref{assumption:transitionbound} and the last step is from Eq.~\ref{eq:vnormsimple}. This proves the second part of the lemma as well. 

\end{proof}

Let us define a new set $\mathcal{C}$ from $C^t_w(\mathcal{S})$ as follows: 
\begin{equation}
    \mathcal{C} \triangleq \Pi_{t=0}^\infty C^t_w(\mathcal{S})
\end{equation}

Also, let us assign the following metric to $\mathcal{C}$: 
    \begin{equation}
        \rho(\boldsymbol{u}, \boldsymbol{v}) \triangleq \sum_{t=0}^\infty \sigma^{-t} || u_t - v_t||_w
    \end{equation}

where $\sigma > 0$ and the following assumption holds.
\begin{assumption}\label{assumption:variablebound}
The variables $\sigma$, $\alpha$, and $\beta$ satisfy $\alpha \sigma \beta < 1$.
\end{assumption}

This assumption guarantees that the metric $\mathcal{C}$ is complete w.r.t $\rho$, and $\rho(\boldsymbol{u}, \boldsymbol{v}) < \infty$ for all $\boldsymbol{u,v}\in \mathcal{C}$.

Using the operators $\{ T^{\boldsymbol{\mu}}_t \}_{t \geq 0}$, let us define a new operator $T^{\boldsymbol{\mu}}: \mathcal{C} \xrightarrow{} \mathcal{C}$ as follows (for all $t \geq 0$). 
\begin{equation}
    (T^{\boldsymbol{\mu}} \boldsymbol{u})_t  = T^{\boldsymbol{\mu}}_t u_{t+1}
\end{equation}

From Lemma~\ref{lemma:operatorlemma} we know that $T^{\boldsymbol{\mu}}$ is a well-defined operator that maps $\mathcal{C}$ to itself. Also, from Eq.~\ref{eq:contractionequation}, $T^{\boldsymbol{\mu}}$ is a contraction operator on $\mathcal{C}$ with constant of contraction $\alpha \beta \sigma < 1$ (from Assumption~\ref{assumption:variablebound}). Now, $T^{\boldsymbol{\mu}}$ will have a unique fixed point by the Banach fixed point theorem in $\mathcal{C}$. 

Now, we move to proving another lemma. Consider a mean field $\boldsymbol{\mu}$, the optimal value function for an agent $j$ can be given by 

\begin{equation}\label{eq:optimalvalue}
    J^{j,\boldsymbol{\mu}}_{*, t}(s^j) = \sup_{\pi^j \in \Pi^j} \E^{\pi^j} \Big[ \sum_{t=0}^\infty \beta^t r^j(s^j_t, a^j_t, \mu_t) | s^j_0 = s^j \Big]. 
\end{equation}

Let $\boldsymbol{J}^{j,\boldsymbol{\mu}}_* \triangleq (J^{j, \boldsymbol{\mu}}_{*, t})_{t \geq 0}$ denote the optimal value function for an agent $j \in \{ 1, \ldots, N \}$. This function is guaranteed to be continuous by Assumption \ref{assumption:rewardfunction}. 

\begin{lemm}
For any $\boldsymbol{\mu}$, the optimal point $\boldsymbol{J}^{j,\boldsymbol{\mu}}_*$, for an agent $j \in \{ 1, \ldots, N \}$, belongs to the Banach space $\mathcal{C}$. 
\end{lemm}

\begin{proof}
Let $\pi$ be a Markov policy. At any time $t \geq 0$, we have 
\begin{equation}
    \begin{array}{l}
         J^{j,\boldsymbol{\mu}}_{*,t}(s)
         \\ \\
         = \sum_{k=t}^\infty \beta^{k-t} \E^\pi(r^j(s^j_k, a^j_k, \mu_{k}) | s^j_t = s^j)
         \\ \\
         \leq  \sum_{k=t}^\infty \beta^{k-t} M_k \E^\pi( v(s_k) | s^j_t = s^j)
         \\ \\
         \leq  \sum_{k=t}^\infty \beta^{k-t} M_k \alpha^{k-t} v(s^j)
         \\ \\
          = L_t v(s^j)
    \end{array}
\end{equation}

The second step is from Assumption~\ref{assumption:rewardbound}. The third step is from Assumption~\ref{assumption:transitionbound}. Hence, the function $\boldsymbol{J}^{j,\boldsymbol{\mu}}_{*,t}$ belongs to the set $C^t_w(\mathcal{S})$ from the Eq.~\ref{eq:vnormbound}. 

\end{proof}




Let us define another set $\mathcal{D}$ as 

\begin{equation}
    \mathcal{D} \triangleq \Pi^\infty_{t=0} \mathcal{T}^t_w (\mathcal{S} \times \mathcal{A}^j)
\end{equation}

Before proving further results, we restate a result proved previously in non-homogeneous stochastic processes \citep{hinderer1970decision}. We will adapt the result to pertain to an agent $j \in \{1, \ldots, N \}$. For any $E$-valued random element $x$, we use the notation $\mathcal{L}(x)$ to denote the distribution of $x$ (i.e. $\mathcal{L}(x) \in P(E)$). 

Consider a variable $\boldsymbol{\nu} \in \mathcal{D}
$. From the operator $T^{\boldsymbol{\mu}}$, we will define another operator $T^{\boldsymbol{\nu}}$, where the relation is such that $\mu_t = \nu_{t,1}$. Recall that the subscript here refers to the state marginal of $\nu$ (refer Eq.~\ref{eq:thedefinitionofT}). From the result of $T^{\boldsymbol{\mu}}$, we know that the operator $T^{\boldsymbol{\nu}}$ is well-defined and has a unique fixed point. Now, we can state the following result.

\begin{lemm}\label{lemma:fixedpointoptimality}
For an agent $j \in \{ 1, \ldots, N \}$, for any $\boldsymbol{\nu} \in \mathcal{D}$, the collection of value functions $\boldsymbol{J}^{j, \boldsymbol{\nu}}_*$ is the unique fixed point of the operator $T^{\boldsymbol{\nu}}$. Furthermore, $\pi^j \in M$ is optimal if and only if 
\begin{equation}\label{eq:defofnu}
    \begin{array}{l}
    \nu^{\pi^j}_t(\{ (s^j_t,a^j_t): r^j(s^j_t, a^j_t, \nu_{t,1}) + \\ \\ \beta \bigintss_{\mathcal{S}}
    J^{j, \boldsymbol{\nu}}_{*,t+1}(s^j_{t+1})p(s^j_{t+1}|s^j_t,a^j_t,\nu_{t,1})   = T^\nu_t J^{j, \boldsymbol{\nu}}_{*, t+1}(s^j_t)\}) \\ \\ = 1
    \end{array}
\end{equation}

where $\nu^{\pi^j}_t = \mathcal{L}(s^j_t,a^j_t)$. 
\end{lemm}

\begin{proof}
See \citet{hinderer1970decision}.

\end{proof}

Now, let us define a set-valued mapping $\tau: \mathcal{D} \xrightarrow{} 2^{|\mathcal{P}(\mathcal{S} \times \mathcal{A})^\infty|}$, for an agent $j$ as follows: 

\begin{equation}
\tau(\boldsymbol{\nu}) = C(\boldsymbol{\nu}) \cap B(\boldsymbol{\nu})    
\end{equation}

where 

\begin{equation}
\begin{array}{l}
        C(\boldsymbol{\nu})   \triangleq \{  \boldsymbol{\nu}' \in \mathcal{P}(\mathcal{S} \times \mathcal{A}^j): \nu'_{0,1} = \mu_0 \textrm{ and } \\ \\ \nu'_{t+1, 1}(\cdot) = \bigintss_{\mathcal{S} \times \mathcal{A}^j} p(\cdot|s^j,a^j, \nu_{t,1}) \nu_t(s^j,a^j) \}
        \end{array}
\end{equation}

and 

\begin{equation}
\begin{array}{l}
    B(\boldsymbol{\nu}) \triangleq \Big\{ \boldsymbol{\nu}' \in \mathcal{P}(\mathcal{S} \times \mathcal{A}^j): \forall t\geq0,\\ \\
    \nu'_t(\big\{(s^j_t,a^j_t): r^j(s^j_t,a^j_t,\nu_{t,1}) + \\ \\ \beta \bigintss_{\mathcal{S}} J^{j, \boldsymbol{\nu}}_{*, t+1}(s^j_{t+1}) p(s^j_{t+1}|s^j_t, a^j_t, \nu_{t,1}) \big\} = 1
    \Big\}
    \end{array}
\end{equation}

Now in the next result we show that, the image of $\mathcal{D}$ under $\tau$ is contained in $2^\mathcal{D}$.

\begin{proposition}
For any $\nu \in \mathcal{D}$, we have $\tau(\nu) \subset \mathcal{D} $ 
\end{proposition}
\begin{proof}
Fix a $\nu \in \mathcal{D}$. To prove the result, it is sufficient to prove that $C(\boldsymbol{\nu}) \subset \mathcal{D}$. Let $\boldsymbol{\nu}' \in C(\boldsymbol{\nu})$. We aim to prove by induction that $\nu'_{t,1} \in P^t_v(\mathcal{S})$ for all $t \geq 0$. The claim trivially holds for $t = 0$ as $\nu'_{0,1} = \mu_0$. Assume that the claim holds for $t$ and consider $t+1$. We have 
\begin{equation}
    \begin{array}{l}
    \bigintss_{\mathcal{S}} w(s_t^j)\nu'_{t+1, 1} (s^j_{t+1}) 
    \\ \\
    = \bigintss_{\mathcal{S} \times \mathcal{A}} \bigintss_{\mathcal{S}} w(s^j_{t+1}) p(s^j_{t+1}|s^j_t,a^j_t, \nu_{t,1}) \nu_t(s^j_t,a^j_t)
    \\ \\ 
    \leq \bigintss_{\mathcal{S}} \alpha w(s^j_t) \nu_{t,1}(s^j_t)
    \\ \\
    \leq \alpha^{t+1} M
    \end{array}
\end{equation}

The third step is from Assumption~\ref{assumption:transitionbound} and the last step is from the fact that $\nu_{t,1} \in P^t_v(\mathcal{S})$. Hence, we can conclude that $\nu'_{t+1,1} \in P^{t+1}_v(\mathcal{S})$. 
\end{proof}

Now we can say that $\boldsymbol{\nu} \in \mathcal{D}$ is a fixed point of $\tau$ if $\boldsymbol{\nu} \in \tau(\boldsymbol{\nu})$. The following lemma connects the mean field equilibrium and the fixed points of $\tau$. 

\begin{lemm}\label{lemm:sufficiency}
Suppose the set valued mapping $\tau$ has a fixed point $\boldsymbol{\nu^j} = (\nu^j_t)_{t \geq 0}$ for an agent $j \in \{ 1, \ldots, N\}$. Consider a Markov policy for the agent $j$ as $\pi^j = (\pi^j_t)_{t \geq 0}$, which is obtained by factoring as $\nu^j_t(s^j_t, a^j_t) = \nu^j_{t,1}(s^j_t) \pi^j_t(a^j_t | s^j_t)$, and let $\boldsymbol{\nu}^j_1 = (\nu^j_{t,1})_{t \geq 0}$. Then the pair $(\pi^j,\boldsymbol{\nu}^j_1)$ is a decentralized mean field equilibrium. 
\end{lemm}

\begin{proof}
If $\boldsymbol{\nu}^j \in \tau(\boldsymbol{\nu}^j)$, then the corresponding Markov policy $\pi^j$ satisfies the Eq.~\ref{eq:defofnu} for $\boldsymbol{\nu}^j$. Thus, by Lemma~\ref{lemma:fixedpointoptimality}, $\pi^j \in \Phi(\boldsymbol{\nu}^j_1)$. Further, since $\boldsymbol{\nu}^j \in C(\boldsymbol{\nu}^j)$, we have $\Psi(\pi^j) = \boldsymbol{\nu}^j_1$, which completes the proof. 
\end{proof}

From the Lemma~\ref{lemm:sufficiency}, it can be seen that the set valued mapping (operator) $\tau$ having a fixed point is sufficient to guarantee the existence of the decentralized mean field equilibrium. Like several results in centralized multiagent systems \citep{nash20167} and centralized mean field games \citep{lasry2007mean, huang2006large, saldi2018discrete}, we will use the Kakutani's fixed point theorem \citep{kakutani1941generalization}, to guarantee the existence of a fixed point for the operator $\tau$. This theorem requires the set on which the set valued mapping $\tau$ operates to be a non-empty, compact and convex subset of some Euclidean space. Further, the operator $\tau(\boldsymbol{\nu}^j)$ is required to be non-empty and convex for all $\boldsymbol{\nu}^j$. Finally, the operator $\tau$ should have a closed graph. Given these three conditions, we can conclude that $\tau$ has a fixed point using the Kakutani's fixed point theorem.

For the first condition, we need to show that the set on which $\boldsymbol{\nu}^j$ resides is non-empty, compact and convex, i.e. we need to show that $\mathcal{D}$ is non-empty, compact and convex. First, note that the function $w$ can be expressed as a continuous moment function and hence the corresponding set $\mathcal{T}^t_v(\mathcal{S})$ is guaranteed to be compact \citep{hernandez2012discrete}. As a consequence, the set $\mathcal{T}^t_v(\mathcal{S} \times \mathcal{A}^j)$ is tight, since the action space $\mathcal{A}^j$ is compact. Also, since the set $\mathcal{T}^t_v(\mathcal{S} \times \mathcal{A}^j)$ is closed, it is compact.   Therefore, the set $\mathcal{D}$ is also compact. From Assumption~\ref{assumption:transitionbound} and Assumption~\ref{assumption:rewardbound} we can show that a line segment between any two points in $\mathcal{D}$ lies in $\mathcal{D}$ and hence the set $\mathcal{D}$ is convex. These assumptions also guarantee that the set $\mathcal{D}$ is non-empty.

For the third condition, we need to show that $\tau(\boldsymbol{\nu}^j)$ is non-empty and convex for any $\boldsymbol{\nu}^j \in \mathcal{D}$. Now, from Lemma~\ref{lemma:fixedpointoptimality} we know that $B(\boldsymbol{\nu}^j)$ is non-empty and hence $\tau(\boldsymbol{\nu}^j)$ is non-empty. Also, we can show that each of the sets $C(\boldsymbol{\nu}^j)$ and $B(\boldsymbol{\nu}^j)$ is convex (see \citet{saldi2018discrete}) and hence, their intersection is convex. This makes the set $\tau(\boldsymbol{\nu}^j)$ convex.

\begin{lemm}\label{lemm:closedgraph}
Using Assumptions~\ref{assumption:rewardfunction}\textemdash\ref{assumption:variablebound}, the graph of $\tau$, i.e. the set 
\begin{equation}
    Gr(\tau) \triangleq \{ (\boldsymbol{\nu}, \boldsymbol{\mathcal{E}}) \in \mathcal{D} \times \mathcal{D}: \boldsymbol{\mathcal{E}} \in \tau(\boldsymbol{\nu})\}
\end{equation}
is closed.
\end{lemm}
\begin{proof}
See Proposition 3.9 in \citet{saldi2018discrete} for complete proof.
\end{proof}

Now, we are ready to give the final result. 

\begin{lemm}\label{lemm:finallemma}
Using Assumptions~\ref{assumption:rewardfunction}\textemdash\ref{assumption:variablebound}, for an agent $j \in \{1, \ldots, N\}$, there exists a fixed point $\boldsymbol{\nu}^j$ of the set valued mapping $\tau: \mathcal{D} \xrightarrow{} 2^\mathcal{D}$. Then, the pair $(\pi^j, \boldsymbol{\nu}_1^j)$ is a decentralized mean field equilibrium, where $\pi^j$ is the policy of agent $j$ and $\boldsymbol{\nu}_1^j$ is its mean field estimate constructed according to Lemma~\ref{lemm:sufficiency}.
\end{lemm}

\begin{proof}
Using the Lemma~\ref{lemm:closedgraph} and our previous results, we have proved that $\mathcal{D}$ is non-empty, compact and convex. Further, the set valued mapping $\tau$ has a closed graph, is non-empty and convex. Hence, by Kakutani's fixed point theorem \citep{kakutani1941generalization}, $\tau$ has a fixed point. Thus, from Lemma~\ref{lemm:sufficiency} this fixed point is the decentralized mean field equilibrium. 
\end{proof}

\begin{theorem2}
An agent $j \in \{1, \ldots, N\}$ in the DMFG admits a decentralized mean field equilibrium $(\pi^j_*, \mu^j_*) \in \Pi^j \times \mathcal{M}$. 
\end{theorem2}

\begin{proof}
Our result follows from Lemma~\ref{lemm:finallemma}. 
\end{proof}

\newpage

\section{Proof of Theorem \ref{theorem:operatormapping}}\label{sec:proofofoperatorexistence}

In this and the subsequent theorems, we follow the theoretical results and proofs in the work by \citet{anahtarci2019value}, and extend them to the decentralized setting.

First, we will start with some definition for the norms, similar to our previous proofs. Consider an agent $j \in \{ 1, \ldots, N\}$. Let $w: \mathcal{S} \times \mathcal{A}^j \xrightarrow{} \Re$ be a continuous weight function. For any measurable function  $v: \mathcal{S} \times \mathcal{A}^j \xrightarrow{} \Re$, the $w$-norm of $v$ is defined as, 
\begin{equation}
||v||_w \triangleq \sup_{s^j,a^j} \frac{|v(s^j,a^j) |}{w(s^j,a^j)}.
\end{equation}

Also, for any measurable function $u: \mathcal{S} \xrightarrow{} \Re$, the $w_{\max}$-norm is defined as 

\begin{equation}
    ||v||_{w_{\max}} \triangleq \sup_{s^j} \frac{u(s^j)}{w_{\max} (s^j)}. 
\end{equation}

Now we will state a set of assumptions needed for our results. 

\begin{assumption}\label{assumption:rewardboundondistance}
For all agents $j \in \{ 1, \ldots, N \}$, the reward function is continuous, and it satisfies the following Lipschitz bounds (for some Lipschitz constants $L_1$ and $L_2$): 
\begin{equation}
\begin{array}{l}
    ||r^j(\cdot, \cdot, \mu_t) - r^j(\cdot, \cdot, \hat{\mu}_t)||_w \leq  L_1 W_1(\mu_t, \hat{\mu}_t), \quad \forall \mu_t, \hat{\mu}_t
    \end{array}
\end{equation}
where $W_1$ is the Wasserstein distance of order 1. Also, 

\begin{equation}
\begin{array}{l}
    \sup_{(a^j_t, \mu_t) \in \mathcal{A}^j \times P(\mathcal{S})} |r^j(s^j_t,a^j_t,\mu_t) - r(\hat{s}^j_t, a^j_t, \mu_t) | \\ \\ \quad \quad \quad \quad \quad \quad \quad \quad \quad \quad \leq L_2 d_X (s^j_t, \hat{s}^j_t), \quad \forall{s^j_t, \hat{s}^j_t}
\end{array}
\end{equation}

\end{assumption}

\begin{assumption}\label{assumption:transitionboundondistance}
For an agent $j \in \{ 1, \ldots, N \}$, the transition function $p(\cdot | s^j,a^j,\mu)$ is weakly continuous in $(s^j, a^j, \mu)$ and satisfies the following Lipschitz bounds (for some Lipschitz constants $K_1$ and $K_2$): 

\begin{equation}
\begin{array}{l}
\sup_{s^j \in S} W_1(p(\cdot|s^j_t,a^j_t,\mu_t),  p(\cdot|s^j_t, \hat{a}^j_t, \hat{\mu}_t)) \\ \\ \leq K_1 (|| a^j_t - \hat{a}^j_t|| + W_1(\mu_t, \hat{\mu}_t)), \quad \forall \mu_t, \hat{\mu}_t, \forall a^j_t, \hat{a}^j_t.
\end{array}
\end{equation}

Also, we have, 

\begin{equation}
\begin{array}{l}
    \sup_{\mu \in P(S)} W_1(p(\cdot|s^j_t,a^j_t,\mu_t), p(\cdot|\hat{s}^j_t, \hat{a}^j_t, \mu_t)) \\ \\ \leq K_2 (d_X(s^j_t, \hat{s}^j_t) + || a^j_t - \hat{a}^j_t||), \quad  \forall s^j_t, \hat{s}^j_t,  \forall a^j_t, \hat{a}^j_t. 
    \end{array}
\end{equation}

\end{assumption}

\begin{assumption}\label{assumption:actionspacerestriction}
The action space $\mathcal{A}^j$ is convex for all agents $j \in \{ 1, \ldots, N\}$. 
\end{assumption}


\begin{assumption}\label{assumption:globalrewardbound}
For all agents $j$ and all time $t$, there exist non-negative real numbers $M$ and $\alpha$ such that for each $(s^j_t,a^j_t,\mu)~\in~\mathcal{S}\times~\mathcal{A}^j~\times~\mathcal{P}(\mathcal{S})$, we have 
\begin{equation}
    r^j(s^j_t,a^j_t,\mu) \leq M 
\end{equation}

\begin{equation}
    \int_S w_{\max} (s_{t+1}^j) p(s_{t+1}^j|s^j_t,a^j_t,\mu) \leq \alpha w(s^j_t,a^j_t) 
\end{equation}

\end{assumption}

\begin{assumption}\label{assumption:productassumption}
The product of constants $\beta \alpha < 1$. 
\end{assumption}

Let $C(\mathcal{S})$ denote the set of real-valued continuous functions on $\mathcal{S}$. Let $Lip(\mathcal{S})$ denote the set of all Lipshitz continuous function on $\mathcal{S}$, i.e., 

\begin{equation}
Lip(\mathcal{S}) \triangleq \{ g \in C(\mathcal{S}): ||g||_{Lip} < \infty \}    
\end{equation}

\noindent where $||g||_{Lip}$ is defined as 

\begin{equation}
    ||g||_{Lip} \triangleq \sup_{(x, y) \in \mathcal{S} \times \mathcal{S}} \frac{|g(x) - g(y)|}{d_X(x,y)}
\end{equation}

Here $g \in C(\mathcal{S})$. The finiteness of $||g||_{Lip}$ guarantees that $g$ is Lipschitz continuous with constant $||g||_{Lip}$. 



Also, let us define $B(\mathcal{S}, K)$ to be the set of all real valued measurable functions in $\mathcal{S}$ with $w_{\max}$-norm less than $K$. We use $Lip(\mathcal{S}, K)$ to denote the set of all Lipschitz continuous functions with $|| g||_{Lip} < \infty$. 

Finally, we define an operator $F: \mathcal{S} \times Lip(\mathcal{S}) \times \mathcal{P}(\mathcal{S}) \times \mathcal{A}^j \xrightarrow{} \Re$ as  

\begin{equation}
\begin{array}{l}
    F: \mathcal{S} \times Lip(\mathcal{S}) \times \mathcal{P}(\mathcal{S}) \times \mathcal{A}^j \ni (s^j_t, v^j_t ,\mu_t, a^j_t) \\ \\ \xrightarrow{} r^j(s^j_t,a^j_t,\mu_t) + \beta \bigintss_{\mathcal{S}} v^j_t(s^j_{t+1}) p(s^j_{t+1}|s^j_t,a^j_t,\mu_t) \in \Re
    \end{array}
\end{equation}

\begin{assumption}\label{assumption:strongconvexity}
Let $\mathcal{F}$ be the set of non-negative functions in 
\begin{equation}
    Lip \Big(\mathcal{S}, \frac{L_2}{1 - \beta K_2} \Big) \cap B \Big(\mathcal{S}, \frac{M}{1 - \beta \alpha} \Big)
\end{equation}

For an agent $j \in \{ 1, \ldots, N \}$, for any $v \in \mathcal{F}$, $\mu \in \mathcal{P}(\mathcal{S})$, and $s^j \in S$, $F(s^j, v^j, \mu, \cdot)$ is $\rho$-strongly convex; that is $F(s^j,v^j,\mu, \cdot)$ is differentiable and its gradient $\nabla F$ satisfies  

\begin{equation}
\begin{array}{l}
    F(s^j_t,v^j_t,\mu_t, a^j_t)  \geq  F(s^j_t,v^j_t,\mu_t, \hat{a}^j_t) \\ \\ + \nabla F(s^j_t,v^j_t,\mu_t, \hat{a}^j_t)^T \cdot (a^j_t - \hat{a}^j_t) + \frac{\rho}{2} ||a^j_t - \hat{a}_t^j||^2
    \\ \\ 
    \end{array}
\end{equation}

for some $\rho>0$ and for all $a^j_t,\hat{a}_t^j \in \mathcal{A}^j$. 

\end{assumption}

\begin{assumption}\label{assumption:gradientbound}
For an agent $j \in \{ 1, \ldots, N \}$, the gradient $\nabla F(s^j_t,v^j_t,\mu_t,a^j_t): \mathcal{S} \times \mathcal{F} \times \mathcal{P}(\mathcal{S}) \times \mathcal{A}^j \xrightarrow{} \Re$ satisfies the following Lipshitz bound (for some Lipschitz constant $K_F$): 

\begin{equation}
\begin{array}{l}
    \sup_{a^j_t \in A^j} ||\nabla F(s^j_t,v^j_t,\mu_t, a^j_t) - \nabla F(\hat{s}^j_t, \hat{v}^j_t, \hat{\mu}_t, a^j_t)|| \\ \\ \leq K_F (d_X(s^j_t, \hat{s}^j_t) + ||v^j_t - \hat{v}^j_t||_{w_{\max}} + W_1 (\mu^j_t, \hat{\mu}_t^j)) \\ \\ 
    \end{array}
\end{equation}

for every $s^j_t, \hat{s}^j_t \in \mathcal{S}; v^j, \hat{v}^j_t \in \mathcal{F}, \mu_t, \hat{\mu}_t \in \mathcal{P}(\mathcal{S})$.

\end{assumption}

 Assumption~\ref{assumption:strongconvexity} and Assumption~\ref{assumption:gradientbound} present a constraint on the change in the value function during each of the updates. These assumptions are generally considered in literature establishing the existence of Lyapunov functions that guarantee the presence of local and global optimal points (maxima or minima) as the case may be \citep{tanaka2003multiple}.

We also need to assume the following for all the constants,  

\begin{assumption}\label{assumption:kdefinition}
The following relation holds 
\begin{equation}
\begin{array}{l}
    k \triangleq \max \{ \beta \alpha + \frac{K_F}{\rho} K_1, \beta \frac{L_2}{1-\beta K_2} + \\ \\  \quad \quad \quad \quad \quad \quad \quad \Big(\frac{K_F}{\rho} + 1 \Big) K_1 +  K_2 + \frac{K_F}{\rho} \} < 1
    \end{array}
\end{equation}
\end{assumption}

For the rest of the proof we will also apply Assumption~\ref{assumption:limitofbelief}. Consider a mean field $\boldsymbol{\mu}$, the optimal value function for an agent $j$ is given by the Eq.~\ref{eq:optimalvalue}. Using the same procedure as the result in Lemma~\ref{lemma:fixedpointoptimality}, this optimal value function can be shown to be a unique fixed point of a Bellman operator $T_{\boldsymbol{\mu}}$ that is a contraction on the $w_{\max}$-norm with a constant of contraction $\beta \alpha$ (from Assumption~\ref{assumption:productassumption}).  Now we can write the following equation, 

\begin{equation}
\begin{array}{l}
    J^{j, \boldsymbol{\mu}}_{*}(s^j_t) =  \max_{a^j \in \mathcal{A}^j} \Big[ r^j(s^j_t, a^j_t, \mu_t) +  \\ \\ \quad \quad \quad \quad \quad \quad \quad \quad \quad  \beta \bigintss_{\mathcal{S}} J^{j, \boldsymbol{\mu}}_{*}(s^j_{t+1})  p(s^j_{t+1}| s^j_t, a^j_t, \mu_t) \Big] \\ \\ \triangleq T_{\boldsymbol{\mu}} J^{j, \boldsymbol{\mu}}_{*}(s^j_t)
    \end{array}
\end{equation}

Let us assume that this maximization is caused by a policy $\pi^j(a^j|s^j, \mu^j)$, which will be the optimal policy. In other words, 

\begin{equation}
    \begin{array}{l}
  \max_{a^j \in \mathcal{A}^j} \Big[ r^j(s^j_t, a^j_t, \mu_t) +  \\ \\ \quad \quad \quad \quad \quad \quad  \beta \bigintss_{\mathcal{S}} J^{j, \boldsymbol{\mu}}_{*}(s^j_{t+1})  p(s^j_{t+1}| s^j_t, a^j_t, \mu_t) \Big] = \\ \\ r^j(s^j_t, \pi^j(a^j_t|s^j_t, \mu^j_t), \mu_t) \\ \\ \quad \quad \quad \quad \quad  + \beta \bigintss_{\mathcal{S}} J^{j, \boldsymbol{\mu}}_{*}(s^j_{t+1})p(s^j_{t+1} | s^j_t, a^j_t, \mu_t) 
    \end{array}
\end{equation}

The objective is to obtain this optimal policy. Towards the same we use the $Q$-functions, where the optimal $Q$-function can be defined as 

\begin{equation}\label{eq:Qfunctiondefined}
\begin{array}{l}
    Q^{j,*}_{\boldsymbol{\mu}} (s^j_t, a^j_t) \\ \\ = r^j(s^j_t,a^j_t,\mu_t) + \beta \bigintss_{\mathcal{S}} J^{j, \boldsymbol{\mu}}_{*} p(s^j_{t+1}|s^j_t,a^j_t,\mu_t)
    \end{array}
\end{equation}

The optimal value function given by $J^{j, \boldsymbol{\mu}}_{*}$ is the maximum of the $Q$-function given by $Q^{j,*}_{\boldsymbol{\mu}, \max}$. Hence, we can rewrite Eq.~\ref{eq:Qfunctiondefined} as follows 

\begin{equation}
    \begin{array}{l}
    Q^{j,*}_{\boldsymbol{\mu}} (s^j_t, a^j_t)  = r^j(s^j_t,a^j_t,\mu_t) \\ \\ \quad \quad \quad \quad \quad  + \beta \bigintss_{\mathcal{S}} Q^{j,*}_{\boldsymbol{\mu}, \max_{a^j}}(s^j_{t+1}, a^j) p(s^j_{t+1}|s^j_t,a^j_t,\mu_t) \\ \\ 
    \triangleq H Q^{j,*}_{\boldsymbol{\mu}} (s^j_t,a^j_t)
    \end{array}
\end{equation}

Here, the operator $H$ is the optimality operator for the $Q$-functions. Our objective is to prove that this operator is a contraction and has a unique fixed point given by $Q_*$. 

First, let us define a set of all bounded $Q$-functions for an agent $j \in \{ 1, \ldots, N\}$: 

\begin{equation}
    \begin{array}{l}
        \mathcal{C} \triangleq 
        \Big \{ Q^j:\mathcal{S} \times
        \mathcal{A}^j \xrightarrow{} [0, \infty); ||Q^j_{\max}||_w \leq \frac{M}{1 - \beta \alpha}  
        \\ \\ \textrm{ and } \\ \\ 
        ||Q^j_{\max}||_{Lip} \leq \frac{L_2}{1 - \beta K_2} 
        \Big \}
    \end{array}
\end{equation}

From Assumption~\ref{assumption:limitofbelief}, Assumption~\ref{assumption:strongconvexity} and Assumption~\ref{assumption:gradientbound}, we know that there exists a unique maximiser $\pi^j(s^j_t, Q^j_t, \mu^j_t)$ for the function $F$ represented as,

\begin{equation}
    \begin{array}{l}
      r^j(s^j_t,a^j_t,\mu_t) \\ \\ + \beta \bigintss_{\mathcal{S}} Q^j_{\boldsymbol{\mu}, \max_{a^j}}(s^j_{t+1}, a^j) p(s^j_{t+1}|s^j_t,a^j_t,\mu_t) \\ \\  = F(s^j_t, Q^j_{\boldsymbol{\mu}, \max}, \mu_t, a^j_t).
    \end{array}
\end{equation}. 

This maximser makes the gradient of $F$, 0. 

\begin{equation}
    \nabla F(s^j_t, Q^j_{\boldsymbol{\mu}, \max}, \mu_t, \pi^j(s^j_t, Q^j_t, \mu^j_t))= 0 
\end{equation}

This shows that the maximiser for the operator $H_2$ in Eq.~\ref{eq:DMFGoperatordefinition} is unique, under the considered assumptions. Now, we are ready to prove the required theorem. 

\begin{theorem2}
The decentralized mean field operator $H$ is well-defined, i.e. this operator maps $\mathcal{C} \times \mathcal{P}(\mathcal{S})$ to itself. 
\end{theorem2}

\begin{proof}

We first apply Assumption~\ref{assumption:limitofbelief}. Now, consider the decentralized mean field operator as defined in Eq.~\ref{eq:DMFGoperatordefinition}. To prove the given theorem, we know that $H_2(Q^j, \boldsymbol{\mu}) \in \mathcal{P}(\mathcal{S})$. We need to prove that $H_1(Q^j, \boldsymbol{\mu}) \in \mathcal{C}$. Let us consider an element $(Q^j, \boldsymbol{\mu}) \in \mathcal{C} \times \mathcal{P}(\mathcal{S})$. Then we have, 

\begin{equation}
    \begin{array}{l}
    \sup_{s^j_t,a^j_t} \frac{\Big |H_1(Q^j, \boldsymbol{\mu}) (s^j_t,a^j_t) \Big|}{w(s^j_t,a^j_t)} \\ \\ 
    = \sup_{s^j_t,a^j_t} \frac{\Big|r^j(s^j_t,a^j_t,\mu) + \beta \bigintss_{\mathcal{S}} Q^j_{\max}(s^j_{t+1}, a^j, \mu^j_{t+1})p(s^j_{t+1}|s^j_t,a^j_t,\mu_t) \Big |}{w(s^j_t,a^j_t)}
    \\ \\
    \leq \sup_{s^j_t,a^j_t} \frac{ \Big|r^j(s^j_t,a^j_t,\mu) \Big|}{w(s^j_t,a^j_t)}  \\ \\ \quad \quad  + \beta \sup_{s^j_t,a^j_t} \frac{\Big| \bigintss_{\mathcal{S}} Q^j_{\max}(s^j_{t+1}, a^j_t,\mu_{t+1} )p(s^j_{t+1}|s^j_t,a^j_t,\mu_t) \Big|}{w(s^j_t,a^j_t)}
    \\ \\ 
    \leq M  \\ \\ + \beta ||Q^j_{\max}||_{w_{\max}} \sup_{s^j_t,a^j_t} \frac{\Big| \bigintss_{\mathcal{S}} w_{\max}(s^j_{t+1})p(s^j_{t+1}|s^j_t,a^j_t,\mu_t) \Big|}{w(s^j_t,a^j_t)}
    \\ \\ 
    \leq M + \beta \alpha ||Q^j_{\max}||_w
    \\ \\ 
    \leq M + \beta \alpha \frac{M}{1 - \beta \alpha} 
    \\ \\ 
    = \frac{M}{1 - \beta \alpha}
    \end{array}
\end{equation}

We used the Assumption~\ref{assumption:globalrewardbound} and the fact that $||Q^j||_{w_{\max}} \leq ||Q^j||_w$.

Also, we have, 

\begin{equation}
    \begin{array}{l}
        | H_1(Q^j, \boldsymbol{\mu})_{\max}(s^j_t) - H_1 (Q^j, \boldsymbol{\mu})_{\max}(\hat{s}^j_t)|  
        \\ \\
        = \max_{a^j \in \mathcal{A}^j} [r^j(s^j_t,a^j_t,\mu_t) \\ \\  \quad \quad + \beta \bigintss_{\mathcal{S}} Q^j_{\max}(s^j_{t+1}, a^j, \mu^j_{t+1})
        p(s^j_{t+1}|s^j_t,a^j_t,\mu_t)] \\ \\  \quad \quad - \max_{a^j \in \mathcal{A}^j}[r^j(\hat{s}^j_t,a^j_t,\mu_t)   \\ \\ \quad \quad + \beta \bigintss_{\mathcal{S}} Q^j_{\max}(s^j_{t+1}, a^j, \mu^j_{t+1})
        p(s^j_{t+1}|\hat{s}^j_t,a^j_t,\mu_t)]
        \\ \\
        \leq \sup_{a^j \in \mathcal{A}^j} \Big| r^j(s^j_t,a^j_t,\mu_t)  - r^j(\hat{s}^j_t,a^j_t,\mu_t) \Big|  \\ \\ + \beta \sup_{a^j \in A^j} \Big|\bigintss_{\mathcal{S}} Q^j_{\max}(s^j_{t+1}, a^j, \mu^j_{t+1})
        p(s^j_{t+1}|s^j_t,a^j_t,\mu_t) \\ \\  \quad \quad \quad - \bigintss_{\mathcal{S}} Q^j_{\max}(s^j_{t+1}, a^j, \mu^j_{t+1})
        p(s^j_{t+1}|\hat{s}^j_t,a^j_t,\mu_t) \Big|
        \\  \\
        \leq L_2 d_X(s, \hat{s}) + \beta K_2 ||Q^j_{\max}||_{Lip} d_X(s, \hat{s})
        \\  \\ 
        \leq \frac{L_2}{1 - \beta K_2} d_X(s, \hat{s})
    \end{array}
\end{equation}

Here we are using the Assumption~\ref{assumption:rewardboundondistance} and Assumption~\ref{assumption:transitionboundondistance}. This proves that $H_1(Q^j, \mu) \in \mathcal{C}$ and concludes our proof. 
\end{proof}

\newpage

\section{Proof of Theorem \ref{theorem:contractionofH}}\label{sec:contractionofH}

\begin{theorem2}

Let $\mathcal{B}$ represent the space of bounded functions in $\mathcal{S}$. Then the mapping $H: \mathcal{C} \times \mathcal{P}(\mathcal{S}) \xrightarrow{} \mathcal{C} \times \mathcal{P}(\mathcal{S})$ is a contraction in the norm of $\mathcal{B}(\mathcal{S})$. 

\end{theorem2}

In this section, we will continue to use the set of assumptions and definitions we introduced in Appendix~\ref{sec:proofofoperatorexistence}. 

\begin{proof}

Fix any $(Q^j, \mu)$ and $(\hat{Q}^j, \hat{\mu})$ in $\mathcal{C} \times \mathcal{P}(\mathcal{S})$, let us consider, 

\begin{equation}\label{eq:H1distancebound}
    \begin{array}{l}
         ||H_1(Q^j, \mu) - H_1 (\hat{Q}^j, \hat{\mu}) ||_w
         \\ \\ 
         = \sup_{s^j_t,a^j_t} \\ \\  \Big [ \frac{r^j(s^j_t,a^j_t,\mu_t) + \beta \bigintss_{\mathcal{S}} Q^j_{\max}(s^j_{t+1}, a^j, \mu^j_{t+1}) p(s^j_{t+1}|s^j_t,a^j_t,\mu_t)} {w(s^j_t,a^j_t)} \\ \\  \quad \quad \quad 
         \frac{- r^j(s^j_t,a^j_t,\hat{\mu}_t) - \beta \bigintss_{\mathcal{S}} \hat{Q}^j_{\max}(s^j_{t+1}, a^j, \mu^j_{t+1})p(s^j_{t+1}|s^j_t,a^j_t,\hat{\mu}_t)}{w(s^j_t,a^j_t)} \Big] 
         \\ \\ 
         \leq \sup_{s^j_t,a^j_t} \frac{|r^j(s^j_t,a^j_t,\mu_t) - r^j(s^j_t,a^j_t, \hat{\mu}_t) |}{w(s^j_t,a^j_t)} \quad \quad \quad \\ \\ +  \beta  \sup_{s^j_t,a^j_t} \frac{\Big| \bigintss_{\mathcal{S}} Q^j_{\max}(s^j_{t+1}, a^j, \mu^j_{t+1}) p(s^j_{t+1}|s^j_t,a^j_t,\mu_t)}{w(s^j_t,a^j_t)} \\ \\  \quad \quad \quad \frac{- \bigintss_{\mathcal{S}} \hat{Q}^j_{\max}(s^j_{t+1}, a^j, \mu^j_{t+1})p(s^j_{t+1}|s^j_t,a^j_t,\hat{\mu}_t) \Big|}{w(s^j_t,a^j_t)} 
         \\  \\

         \leq L_1 W_1(\mu, \hat{\mu}) \\ \\ \quad \quad + \beta  \sup_{s^j_t,a^j_t} \frac{\Big| \bigintss_{\mathcal{S}} Q^j_{\max}(s^j_{t+1}, a^j, \mu^j_{t+1}) p(s^j_{t+1}|s^j_t,a^j_t,\mu_t)}{w(s^j_t,a^j_t)} 
         \\ \\  \quad \quad \quad \frac{- \bigintss_{\mathcal{S}} \hat{Q}^j_{\max}(s^j_{t+1}, a^j, \mu^j_{t+1})p(s^j_{t+1}|s^j_t,a^j_t,\mu_t) \Big|}{w(s^j_t,a^j_t)} 
         \\ \\ 
          \quad \quad + \beta  \sup_{s^j_t,a^j_t} \frac{\Big| \bigintss_{\mathcal{S}} \hat{Q}^j_{\max}(s^j_{t+1}, a^j, \mu^j_{t+1}) p(s^j_{t+1}|s^j_t,a^j_t,\mu_t)\Big |}{w(s^j_t,a^j_t)} 
         \\ \\  \quad \quad \quad \frac{- \bigintss_{\mathcal{S}} \hat{Q}^j_{\max}(s^j_{t+1}, a^j, \mu^j_{t+1})p(s^j_{t+1}|s^j_t,a^j_t,\hat{\mu}_t) \Big|}{w(s^j_t,a^j_t)} 
         \\ \\

         \leq L_1 W_1(\mu, \hat{\mu}) +  \beta || Q^j_{\max} - \hat{Q}^j_{\max}||_{w_{\max}} \\ \\ \quad \quad \sup_{s^j_t,a^j_t} \frac{\Big| \bigintss_{\mathcal{S}} w_{\max}(s^j_{t+1}) p(s^j_{t+1}|s^j_t,a^j_t,\mu_t)\Big|}{w(s^j_t,a^j_t)} 
         \\ \\ 
         + \beta || \hat{Q}^j_{\max}||_{Lip}  \sup_{s^j_t,a^j_t} \frac{W_1(p(\cdot|s^j_t, a^j_t, \mu_t), p(\cdot| s^j_t, a^j_t, \hat{\mu}_t))}{w(s^j_t,a^j_t)} 
         
         \\ \\ 
        \leq L_1 W_1 (\mu, \hat{\mu}) + \beta \alpha ||Q^j - \hat{Q}^j||_w +  \beta \frac{L_2}{1 - \beta K_2} K_1 W_1 (\mu, \hat{\mu})

    \end{array}
\end{equation}

First we apply Assumption~\ref{assumption:rewardboundondistance}. In the last step we apply Assumption~\ref{assumption:transitionboundondistance} and Assumption~\ref{assumption:globalrewardbound}. 

Now consider the distance between $H_2(Q^j,\boldsymbol{\mu})$ and $H_2(\hat{Q}^j, \hat{\boldsymbol{\mu}})$. First, we will consider the difference between the unique maximiser $\pi^j(s^j,Q^j,\mu^j)$ of $H_1 (Q^j, \boldsymbol{\mu}) (s^j, a^j)$ and the unique maximiser $\pi^j(s^j, \hat{Q}^j, \hat{\mu}^j)$ of $H_1(\hat{Q}^j, \hat{\boldsymbol{\mu}})(s^j,a^j)$ with respect to the action (using the Assumption~\ref{assumption:limitofbelief}). Let us consider the function,

\begin{equation}
\begin{array}{l}
    F: \mathcal{S} \times \mathcal{C} \times \mathcal{P}(\mathcal{S}) \times \mathcal{A}^j \ni (s^j_t,v^j_t,\mu_t, a^j_t) \\ \\ \xrightarrow{} r^j(s^j_t,a^j_t,\mu_t) + \beta \bigintss_{\mathcal{S}} v(s^j_{t+1}) p(s^j_{t+1} | s^j_t,a^j_t,\mu_t) \in \Re
    \end{array}
\end{equation}

This is $\rho$-strongly convex by Assumption \ref{assumption:strongconvexity}, hence it satisfies the following equation \citep{hajek2019statistical}, 

\begin{equation}\label{eq:equationstrongconvexity}
    [\nabla F (s^j_t,v^j_t, \mu_t, a^j_t + b^j_t) - \nabla F(s^j_t,v^j_t,\mu_t, a^j_t)]^T \cdot r^j_t \geq \rho ||b^j_t||^2
\end{equation}

For any $a^j_t, b^j_t \in \mathcal{A}^j$ and for any $s^j_t \in \mathcal{S}$, we can use the notation $a^j_t = \pi^j(s^j_t,Q^j_t,\mu^j_t)$ and $b^j = \pi^j(s^j_{t+1},\hat{Q}^j_t,\hat{\mu}^j_t) - \pi^j(s^j_t,Q^j_t,\mu^j_t)$. 

Now, $a^j_t = \pi^j(s^j_t,Q^j_t,\mu^j_t)$ is the unique maximiser of the strongly convex function $F(s^j_t, Q_{max,t}^j, \mu_t, \cdot)$, we have 

\begin{equation}
    \nabla F(s^j_t, Q_{max,t}^j, \mu_t, \pi^j(s^j_t,Q^j_t,\mu^j_t)) = 0 
\end{equation}

Also, the term $a^j_t + b^j_t = \pi^j(s^j_{t+1},\hat{Q}^j_t,\hat{\mu}^j_t)$ is the unique maximiser of $F(s^j_{t+1}, 
   \hat{Q}_{max,t}^j, \hat{\mu}_t, \cdot)$. Now, using Assumption~\ref{assumption:strongconvexity} and Eq.~\ref{eq:equationstrongconvexity} we have, 

\begin{equation}\label{eq:boundonpho}
\begin{array}{l}
   - \nabla F(s^j_{t+1}, 
   \hat{Q}_{max,t}^j, \hat{\mu}_t, a^j_t)^T \cdot b^j_t \\ \\ \quad \quad \quad = - \nabla F(s^j_{t+1}, 
   \hat{Q}_{max,t}^j, \hat{\mu}_t, a^j_t)^T \cdot b^j_t  \\  \quad \quad \quad + \nabla F(s^j_{t+1}, \hat{Q}_{max,t}^j, \hat{\mu}_t, a^j_t + b^j_t)^T \cdot b^j_t \geq \rho ||b^j_t||^2
   \end{array}
\end{equation}

Similarly, using the Assumption~\ref{assumption:gradientbound} we also have, 

\begin{equation}
    \begin{array}{l}
         - \nabla F(s^j_{t+1}, 
   \hat{Q}_{max,t}^j, \hat{\mu}_t, a^j_t)^T \cdot b^j_t \\ \\ =  - \nabla F(s^j_{t+1}, 
   \hat{Q}_{max,t}^j, \hat{\mu}_t, a^j_t)^T \cdot b^j_t  \\ \quad \quad \quad  + \nabla F(s^j_t, 
   Q_{max,t}^j, \mu_t, a^j_t)^T \cdot b^j_t
        \\ \\
        \leq ||b^j_t|| || \nabla F(s^j_t, 
   Q_{max,t}^j, \mu_t, a^j_t) \\ \\ \quad \quad - \nabla F(s^j_{t+1}, 
   \hat{Q}_{max,t}^j, \hat{\mu}_t, a^j_t)||
        \\ \\ 
        \leq K_F ||b^j_t||(d_X(s^j_t, s^j_{t+1}) + ||Q_{max,t}^j - \hat{Q}_{max,t}^j||_{w_{max}}  \\ \\ \quad \quad \quad + W_1(\mu_t, \hat{\mu}_t) )
        \\ \\
        \leq K_F ||b^j_t|| (d_X(s^j_t, s^j_{t+1}) + ||Q_{t}^j - \hat{Q}_{t}^j||_{w}  \\ \\ \quad \quad \quad + W_1(\mu_t, \hat{\mu}_t) )
    \end{array}
\end{equation}

Therefore, from the above two equations, 

\begin{equation}\label{eq:policydifference}
\begin{array}{l}
    \pi^j(s^j_{t+1},\hat{Q}^j_t,\hat{\mu}^j_t) - \pi^j(s^j_t,Q^j_t,\mu^j_t) \\ \\ \leq \frac{K_F}{\rho} (d_X(s^j_{t+1},s^j_{t}) + ||Q_{t}^j - \hat{Q}_{t}^j||_{w} + W_1(\mu_t, \hat{\mu}_t) )
    \end{array}
\end{equation}

Now, consider the $W_1$ distance between $H_2(Q^j, \boldsymbol{\mu})$ and $H_2(\hat{Q}^j, \hat{\boldsymbol{\mu}})$.

\begin{equation}\label{eq:H2distancebound}
    \begin{array}{l}
    W_1(H_2(Q^j,\boldsymbol{\mu}), H_2(\hat{Q}^j, \boldsymbol{\hat{\mu}})) \\ \\
    = \sup_{||g||_{Lip} \leq 1}  \Big|\bigintss_{\mathcal{S} \times \mathcal{A}^j} \bigintss_{\mathcal{S}} g(s^j_{t+1}) \\ \\ \quad \quad \quad p\big(s^j_{t+1}|s^j_t, \pi^j(s^j_t,Q^j_t, \mu^j_t), \mu_t \big) \mu_t(s^j_t)
     \\ \\ \quad \quad 
    - \bigintss_{\mathcal{S} \times \mathcal{A}^j} \bigintss_{\mathcal{S}} g(s^j_{t+1}) \\ \\ \quad \quad \quad p\big(s^j_{t+1}|s^j_t, \pi^j(s^j_t, \hat{Q}^j_t, \hat{\mu}^j_t), \hat{\mu}_t\big) \hat{\mu}_t(s^j_t) \Big|
    \\ \\
    \leq 
       \sup_{||g||_{Lip} \leq 1}  \Big|\bigintss_{\mathcal{S} \times \mathcal{A}^j} \bigintss_{\mathcal{S}} g(s^j_{t+1}) \\ \\ \quad \quad \quad p\big(s^j_{t+1}|s^j_t, \pi^j(s^j_t,Q^j_t, \mu^j_t), \mu_t\big) \mu_t(s^j_t)\\ \\
     \quad \quad - \bigintss_{\mathcal{S} \times \mathcal{A}^j} \bigintss_{\mathcal{S}} g(s^j_{t+1}) \\ \\ \quad \quad \quad p\big(s^j_{t+1}|s^j_t, \pi^j(s^j_t, \hat{Q}^j_t, \hat{\mu}^j_t), \hat{\mu}_t\big) \mu_t(s^j_t) \Big|
    \\ \\
      + \sup_{||g||_{Lip} \leq 1}  \Big|\bigintss_{\mathcal{S} \times \mathcal{A}^j} \bigintss_{\mathcal{S}} g(s^j_{t+1}) \\ \\ \quad \quad \quad p\big(s^j_{t+1}|s^j_t, \pi^j(s^j_t,\hat{Q}^j_t, \hat{\mu}^j_t), \hat{\mu}_t\big) \mu_t(s^j_t)\\ \\ \quad \quad  
    - \bigintss_{\mathcal{S} \times \mathcal{A}^j} \bigintss_{\mathcal{S}} g(s^j_{t+1}) \\ \\ \quad \quad \quad  p\big(s^j_{t+1}|s^j_t, \pi^j(s^j_t, \hat{Q}^j_t, \hat{\mu}^j_t), \hat{\mu}_t\big) \hat{\mu}_t(s^j_t) \Big|
    \\ \\

  \overset{(1)}{\leq} 
       \bigintss_{\mathcal{S} \times 
       \mathcal{A}^j}\sup_{||g||_{Lip} \leq 1} \Big| \bigintss_{\mathcal{S}} g(s^j_{t+1}) \\ \\ \quad \quad \quad  p\big(s^j_{t+1}|s^j_t, \pi^j(s^j_t,Q^j_t, \mu^j_t), \mu_t\big) \mu_t(s^j_t)\\ \\
    - \bigintss_{\mathcal{S} \times \mathcal{A}^j} \bigintss_{\mathcal{S}} g(s^j_{t+1})  p\big(s^j_{t+1}|s^j_t, \pi^j(s^j_t, \hat{Q}^j_t, \hat{\mu}^j_t), \hat{\mu}_t\big)\Big| \\ \\ \quad \quad \quad  \mu_t(s^j_t) 
       + (K_2 + \frac{K_F}{\rho}) W_1(\mu_t, \hat{\mu_t})
      \\ \\ 
      \leq 
    \bigintss_{\mathcal{S} \times 
       \mathcal{A}^j}  W_1 \Big( p\big(\cdot|s^j_t, \pi^j(s^j_t,Q^j_t, \mu^j_t), \mu_t\big), \\ \\ \quad  \quad \quad p\big(\cdot| s^j_t, \pi^j(s^j_t, \hat{Q}^j_t, \hat{\mu}^j_t), \hat{\mu}_t \big) \Big)\mu_t(s^j_t)  
       \\ \\ \quad \quad \quad 
       + \Big( K_2 + \frac{K_F}{\rho} \Big) W_1(\mu_t, \hat{\mu_t})

    \\ \\ 
   \overset{(2)}{\leq}  
    \frac{K_F}{\rho} K_1 \Big ( ||Q^j_t - \hat{Q}^j_t||_w + W_1(\mu_t, \hat{\mu}_t) \Big) \\ \\ \quad \quad
    + K_1 W_1(\mu_t, \hat{\mu}_t)
    + \Big( K_2 + \frac{K_F}{\rho} \Big) W_1(\mu_t, \hat{\mu_t})
   \end{array} 
\end{equation}

In the above derivation, (1) and (2) are obtained from Assumption~\ref{assumption:transitionboundondistance} and Eq.~\ref{eq:policydifference} (also refer \citet{anahtarci2019value}). Combining Eq.~\ref{eq:H1distancebound} and Eq.~\ref{eq:H2distancebound}, and applying Assumption~\ref{assumption:kdefinition}, the required result is proved. The constant of contraction is $k$ defined in Assumption~\ref{assumption:kdefinition}.

\end{proof}

\section{Proof of Theorem \ref{theorem:convergencetoequilibrium}}\label{sec:convergencetoequilibrium}

\begin{theorem2}
Let the $Q$-updates in Algorithm~\ref{alg:Qldmfg} converge to $(Q^j_*, \mu_*^j)$ for an agent $j \in \{1, \ldots, N\}$. Then, we can construct a policy $\pi^j_*$ from $Q^j_*$ which is expressed as
\begin{equation*}
    \pi^j_*(s^j) = \arg \max_{a^j \in \mathcal{A}^j} Q^j_*(s^j, a^j, \mu^j_*).
\end{equation*}
Then the pair $(\pi^j_*, \mu_*^j)$ is a DMFE.
\end{theorem2}

\begin{proof}

From the Theorem~\ref{theorem:contractionofH} we know that $(Q^j_*, \mu^j_*)$ is a fixed point of $H$. Using the Assumption~\ref{assumption:limitofbelief}, we consider a $\mu_*$, where $\mu_*(s^j) = \mu^j_*(s^j)$, for all states $s^j \in \mathcal{N}^j$. Hence, we can construct the following equations.  
\begin{equation}
\begin{array}{l}
       Q^{j}_{*} (s^j_t, a^j_t, \mu^j_*)  = r^j(s^j_t,a^j_t,\mu_{*,t}) \\ \\ \quad \quad  \quad  + \beta \bigintss_{\mathcal{S}} Q^{j}_{*, \max}(s^j_{t+1}, a^j, \mu^j_{*, t+1}) p(s^j_{t+1}|s^j_t,a^j_t,\mu_{*,t}) 
    \end{array}
\end{equation}

\begin{equation}\label{eq:constructionofmu}
    \mu^j_{*, t+1}(\cdot) = \int_{\mathcal{S} \times \mathcal{A}^j} p(\cdot| s^j_t, a^j_t, \mu_{*,t}) \pi^j_*(a^j_t|s^j_t, \mu^j_{*,t}) \mu_{*,t}(s^j_t)
\end{equation}

From Eq.~\ref{eq:constructionofmu} and Assumption~\ref{assumption:limitofbelief}, we can construct the mean field estimate of agent$j$, $\boldsymbol{\mu}^j_*$.  Now, the above equations imply that $\pi^j_* \in \Phi(\mu^j_*)$ and $\mu^j_* = \Psi(\pi^j_*)$. Hence, $(\pi^j_*, \mu^j_*)$ is a decentralized mean field equilibrium.

\end{proof}

\section{Differences Between Mean Field Settings}\label{sec:differencesofmeanfieldsettings}
Table~\ref{table:differences} captures the differences between the three mean field frameworks, MFRL, MFG, and DMFG. Particularly, we show that the DMFG is different from other frameworks introduced previously, and it is more practical than prior methods as it relaxes some strong assumptions in them.

There are several disadvantages in using a centralized solution concept like Nash equilibrium (or, by extension, the mean field equilibrium) in multiagent systems. These are listed as follows: 1) Centralized nature of the equilibrium, though practical systems have a decentralized information structure. 2) The equilibrium computation is almost intractable for more than two agents \citep{neumann1928theorie}. 
3) Needs strong assumptions to give theoretical guarantees of convergence of learning systems in general sum games \cite{hu2003nash}, even in the stationary case. 
4) The equilibrium requires agreement between agents even in the competitive case. 5) It is hard to verify this point in practice. 
Our decentralized solution concept (DMFE) is more practical than the Nash equilibrium and the mean field equilibrium, since it does not suffer these limitations noted for the centralized methods. 

\begin{table}
\centering
\renewcommand{\arraystretch}{2}
\begin{tabular}{||p{0.20 \linewidth}  | p {0.18\linewidth}  | p {0.18\linewidth}  | p {0.21
\linewidth} ||} 
 \hline\hline
 Method & MFRL \citet{pmlr-v80-yang18d} & MFG \citet{lasry2007mean} & DMFG (ours) \\ [0.5ex] 
 \hline\hline
 Game formulation  & Stochastic Game & Mean Field Game & Decentralized Mean Field Game \\ 
 \hline
  Reward function & Same for all agents & Same for all agents & Can be different \\ 
 \hline
 Action space  & Discrete only & Can be continuous & Can be continuous \\ 
 \hline
  Action space  & Same for all agents & Same for all agents & Can be different \\ 
 \hline
 Solution Concept & Nash Equilibrium (centralized) & Mean Field Equilibrium (centralized) & Decentralized Mean Field Equilibrium \\
  \hline
 Complexity of obtaining the mean field  & Exponential in agents & Exponential in agents  & Constant in agents \\
 \hline
 Theoretical Guarantees & Needs very strong assumptions & Needs weak assumptions & Needs weak assumptions \\
 \hline
 Number of agents & Should be large but finite & Can be infinite in the limit & Can be infinite in the limit \\
 \hline
 Mean Field Information & Global & Global & Local \\
  \hline
 Agents & Identical and homogeneous & Identical and homogeneous & Non-identical\\
 \hline
 Policy & Stationary & Can be non-stationary & Can be non-stationary\\
 \hline
 State Observable & Global & Local & Local \\ [1ex] 
 \hline
 \hline
\end{tabular}
\caption{Table to capture the differences between Mean Field Reinforcement Learning (MFRL), Mean Field Games (MFG), and Decentralized Mean Field Games (DMFG).}
\label{table:differences}
\end{table}

\section{Mean Field Modelling in DMFG-QL}\label{sec:meanfieldmodelling}

In this section, we plot the mean square error between the estimated mean field action ($\mu^{j,a}_t$) (from the neural network representing Eq.~\ref{eq:Dupdatemeana}) and the true observed local mean  field action ($\hat{\mu}^{j,a}_t$) to show that the true local mean field action can be accurately modelled by the DMFG-QL algorithm. We will use the Battle game for this illustration. 

\begin{figure}[h]
\centering
\includegraphics[width=0.45\textwidth]{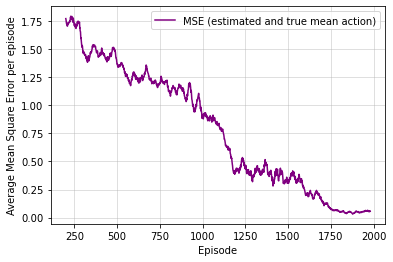}
\caption{This figure shows the average (per agent)  mean square error (MSE) between the true local mean field action and the estimated mean field action from the neural network in each episode (sum of MSE in the 500 steps) of the Battle game training experiment. The results show that the MSE steadily reduces, and hence DMFG-QL is able to accurately model the true mean field. The result shows the average of 10 runs (negligible standard deviation). }
\label{fig:mse}
\end{figure}

Fig.~\ref{fig:mse} shows the mean square errors (MSE) between the true mean field action and the estimated mean field action, averaged per agent, for each episode of the Battle game. The MSE converges to a small value close to 0 during training. This shows that the DMFG-QL algorithm can accurately model the mean field, which contributes to its superior performance in many of our experimental domains. As stated in Section~\ref{sec:algorithms}, DMFG-QL algorithm formulates best responses to the current mean field action. This contrasts with approaches such as MFQ that formulate best responses to the previous mean field, which leads to a loss in performance. We start the plot in Fig~\ref{fig:mse} from episode 200 (instead of 0), since the first few episodes have a very high MSE, that simply skews the axis of the resulting graph.

\section{Algorithm Pseudo-Code}\label{sec:algorithmimplementation}

Algorithm~\ref{alg:POGMFGQL} provides the pseudo-code of the DMFG-QL algorithm. Here, the neural networks are used as the function approximator for the $Q$-function. 
In addition, we include the use of a second target network and a replay buffer for training, as done in the popular Deep $Q$-learning (DQN) algorithm \citep{mnih2015human}. Also, similar to \citet{pmlr-v80-yang18d}, we specify that the policies satisfy the GLIE assumption and hence we use the max operator over the $Q$-value of the next state (i.e. choose the action that maximizes the Q value) to calculate the temporal difference (T.D.) error instead of maintaining a distinct value function as in Eq.~\ref{eq:Dvalueupdate}. 

Algorithm~\ref{alg:POGMFGAC} provides the pseudo-code of the DMFG-AC algorithm. This algorithm uses the policy as the actor, and the value function as the critic, as is common in practice \citep{konda1999actor}. We choose to use a stochastic actor that does not depend on the mean field. This is done to provide an algorithm that can work independent of the mean field during execution. Since the critic (not used during execution) can be dependent on the mean field, we additionally parameterize the value function with the mean field. The modified updates can be seen in Algorithm~\ref{alg:POGMFGAC}. For both Algorithm~\ref{alg:POGMFGQL} and Algorithm~\ref{alg:POGMFGAC} we initialize the estimated mean field distribution to a uniform distribution (arbitrarily).

\begin{algorithm}[ht!]
\caption{Q-learning for Decentralized Mean Field Games (DMFG-QL)}
\label{alg:POGMFGQL}
\begin{algorithmic}[1] 

\STATE Initialize the $Q$-functions (parameterized by weights) $ Q_{\phi^j}, Q_{\phi^j_{\_}}$, for all agents $j \in \{ 1, \ldots, N \}$
\STATE Initialize the mean field estimate (parameterized by weights) $\mu^{j,a}_{\eta}$ for each agent $j \in \{1, \ldots, N \}$
\STATE Initialize the estimated mean field $\mu^{j,a}_{0, \eta}$ for each $j$ to a uniform distribution
\STATE Initialize the total steps (T) and total episodes (E)
\STATE Obtain the current state $s^j_t$
\WHILE {Episode $<$ E}
\WHILE{Step $<$ T}

\STATE For each agent $j$, obtain action $a^j_t$ from the policy induced by $Q_{\phi^j}$ with the current estimated mean action $\mu^{j,a}_{t, \eta}$ and the exploration rate $\hat{\beta}$ according to Eq.~\ref{eq:Dupdatepolicy} 
\STATE Execute the joint action $\textbf{a}_t = [a^1_t,\ldots, a_t^N]$. Observe the rewards $\textbf{r}_t = [r^1_t,\ldots, r^N_t]$ and the next state $\textbf{s}_{t+1} = [s_{t+1}^1, \ldots, s_{t+1}^N]$
\STATE For each agent $j$, obtain the current observed local mean action ($\hat{\mu}^{j,a}_t$)
\STATE Update the parameter $\eta$ of the mean field network using a mean square error between the observed mean action $\hat{\mu}^{j,a}_t$ and the current estimated mean action $\mu^{j,a}_{t, \eta}$ for all $j$
\STATE For each $j$, obtain the mean field estimates $\mu^{j,a}_{t+1, \eta}$ for the next state $s_{t+1}^j$ using the mean field network according to Eq.~\ref{eq:Dupdatemeana} 
\STATE For each $j$, store $\langle s^j_t,a^j_t,r^j_t,s_{t+1}^j, \mu^{j,a}_{t, \eta}, \mu^{j,a}_{t+1, \eta}  \rangle$ 
in replay buffer $B$
\STATE Set the next mean field estimate as the current mean field estimate $\mu^{j,a}_{t, \eta}  = \mu^{j,a}_{t+1, \eta} $ and the next state as the current state $s_{t}^j = s_{t+1}^j$ for all agents $j \in \{ 1, \ldots, N\}$
\ENDWHILE
\WHILE{$j$ = $1$ to $N$}
\STATE Sample a minibatch of K experiences  $\langle s^j_t,a^j_t,r^j_t,s_{t+1}^j, \mu^{j,a}_{t, \eta}, \mu^{j,a}_{t+1, \eta}  \rangle$ from $B$
\STATE Set $y^j = r^j_t + \gamma \max_{a^j_{t+1}}Q_{\phi^j_{\_}}(s^j_{t+1}, a^j_{t+1},\mu^{j,a}_{t+1, \eta})$ according to Eq.~\ref{eq:DMFQ}
\STATE Update the Q network by minimizing the loss $L(\phi^j) = \frac{1}{K} \sum (y^j - Q_{\phi^j}(s^j_{t}, a^j_{t},\mu^{j,a}_{t, \eta}))^2$
\ENDWHILE
\STATE Update the parameters of the target network for each agent $j$ with learning rate $\tau$; $$\phi^{j}_{\_} \leftarrow \tau \phi^j + (1 - \tau) \phi^j_{\_}$$
\ENDWHILE
\end{algorithmic}
\end{algorithm}

\begin{algorithm}[ht!]
\caption{Actor-Critic for Decentralized Mean Field Game (DMFG-AC)}
\label{alg:POGMFGAC}
\begin{algorithmic}[1] 
\STATE Initialize the $V$-function or critic (parameterized by weights) $ V_{\phi^j}$ and the policy or actor (parameterized by weights) $\pi_{\theta^j}$ for all agents $j \in {1, \ldots, N}$
\STATE Initialize the mean field estimate (parameterized by weights) $\mu^{j,a}_{\eta}$ for each agent $j \in \{1, \ldots, N \}$
\STATE Initialize the estimated mean field $\mu^{j,a}_{0, \eta}$ for each $j$ to a uniform distribution
\STATE Initialize the total episodes (E)
\STATE Obtain the current state $s^j_t$ 
\WHILE {Episode $<$ E}
\STATE For each agent $j$, obtain action $a^j_t$ from the policy $\pi_{\theta^j}$ at the current state $s^j_t$
\STATE Execute the joint action $\textbf{a}_t = [a^1_t,\ldots, a_t^N]$. Observe the rewards $\textbf{r}_t = [r^1_t,\ldots, r^N_t]$ and the next state $\textbf{s}_{t+1} = [s_{t+1}^1, \ldots, s_{t+1}^N]$
\STATE For each agent $j$, obtain the current observed local mean field action ($\hat{\mu}^{j,a}_t$) 
\STATE Update the parameter $\eta$ of the mean field network using a mean square error between the observed mean action $\hat{\mu}^{j,a}_t$ and the current estimated mean action $\mu^{j,a}_{t, \eta}$ for all $j$
\STATE For each $j$, obtain the mean field estimates $\mu^{j,a}_{t+1, \eta}$ for the next state $s_{t+1}^j$ using the mean field network according to Eq.~\ref{eq:Dupdatemeana} 
\STATE Set $y^j = r^j_t + \gamma  V_{\phi^j}(s^j_{t+1},\mu^{j,a}_{t+1, \eta})$ as the T.D. target according to Eq.~\ref{eq:DMFQ}
\STATE For each $j$, update the critic by minimizing the loss $L(\phi^j) = (y^j - V_{\phi^j}(s^j_{t}, \mu^{j,a}_{t, \eta}))^2$ 
\STATE For each $j$, update the actor using the log loss 
$\mathcal{J}(\theta^j) = \log \pi_{\theta^j}(a^j_t | s^j_t) L(\phi^j) $
\STATE Set the next mean field estimate as the current mean field estimate $\mu^{j,a}_{t, \eta}  = \mu^{j,a}_{t+1, \eta} $ and the next state as the current state $s_{t}^j = s_{t+1}^j$ for all agents $j \in \{ 1, \ldots, N\}$
\ENDWHILE
\end{algorithmic}
\end{algorithm}

\section{Experimental Details}\label{sec:experimentaldetails}

In this section, we describe each of our game domains in detail, including the reward functions. We also provide the implementation details of our algorithms, especially the hyperparameters. We also discuss the wall-clock times of our algorithmic implementations. Each episode for the Petting Zoo environments has a maximum of 500 steps.  In our implementation of MFQ and MFAC, the agents learn in a decentralized fashion with independent networks and use the previous mean field of the local neighbourhood instead of the global mean field.

The first 5 domains are obtained from the Petting Zoo environment \citep{terry2020pettingzoo}, and the game parameters are mostly left unchanged from those given in \citet{terry2020pettingzoo}. In these games, as agents can die during game play, we use the agent networks saved in the last available episode during training for the execution experiments. Complete details of these domains are given in each of the sub-sections below. 

\subsection{Battle Domain}

This is the first domain from the Petting Zoo environment \citep{terry2020pettingzoo} that is mixed cooperative-competitive. This domain was originally defined in the MAgents simulator \citep{zheng2018magent}. We have two teams of 25 agents, each learning to cooperate against the members of the same team and compete against the members of its opponent team. The agent gets rewarded for attacking and killing agents of the opposing team. At the same time, the agent is penalized for being attacked. In our implementation, each agent learns using its local observation and cannot get global information. Each agent has a view range of a circular radius of 6 units around it. Most rewards are left as defaults, as given in \citet{terry2020pettingzoo}. The agents get a penalty of -0.005 for each step (step reward) and a penalty of -0.1 for being killed. The agents get a reward of 5 for attacking an opponent and a reward of 10 for killing an opponent. There is an attack penalty of -0.1 (penalty for attacking anything). Agents start with a specific number of hitpoints (HP) that are damaged upon being attacked. Agents lose 2 HP when attacked and recover an HP of 0.1 for every step that they are not attacked. They start with 10 HP, and when this HP is completely lost, the agent dies. Each agent can view a circular range of 6 units around it (view range) and can attack in a circular range of 1.5 units around it (attack range). The action space of each agent is a discrete set of 21 values, which corresponds to taking moving and attacking actions in the environment. In this game, we assume that agents can get perfect information about actions of other agents at a distance of 6 units from themselves and no information beyond this point. In the execution phase, a game is considered to be won by a team that kills more opponent agents. 
If both teams kill the same number of agents, the team having the higher cumulative reward is determined as the winner. 

\subsection{Gather Domain}

In this environment, all agents try to capture limited food available in the environment and gain rewards. Each food needs to be attacked before it can be captured (takes 5 attacks to capture food). This is a fully competitive game, where all agents try to outcompete others in the environment and gain more food for themselves. Agents can also kill other agents in the environment by attacking them (just a single attack). Each agent gets a step penalty of -0.01, an attack penalty (penalty for attack) of -2 and a death penalty (punishment for dying) of -20. Also, each agent gets a reward of 20 for attacking food and a reward of 60 for capturing the food. There are a total of 30 agents learning in our environment. The action space of each agent is a set of 33 values. The agents have a view range of 7 and an attack range of 1. All other conditions and rewards are the same as the Battle game.  

\subsection{Combined Arms Domain}

The Combined Arms environment is a heterogeneous large-scale team game with two types of agents in each team. The first type is a ranged agent, which can move fast and attack agents situated further away but has fewer HP. The second type is the melee agent that can only attack close-by agents and move slowly; however, they have more HP. The reward function for the agents is the same as that given in the battle game. The action space of the melee agents is a discrete set of 9 values, while the action space of the ranged agents is a discrete set of 25 values. The actions correspond to moving in the environment and attacking neighbouring agents. The ranged agents have a maximum HP of 3, and the melee agents have a maximum HP of 10. The maximum HP is the limit after which the HP of any agent cannot increase in this environment. Like the battle domain, agents lose 2 HP for each time they are attacked and gain 0.1 HP for each step without being attacked. In our experiments, each team consists of 25 agents, with 15 ranged and 10 melee agents at the start. The view range is 6 and the attack range is 1 for melee agents. The view range is 6 and attack range is 2 for ranged agents. All other conditions and rewards are the same as the Battle game. To create the mean field we choose to use the action space of the type which has the larger number of actions (i.e. we use the action space of the ranged agents). 

\subsection{Tiger-Deer Domain}

In the tiger-deer environment, tigers are the learning agents, cooperating with each other to kill deer in the environment. At least two tigers need to attack a deer together to get high rewards. The tigers start with an HP of 10 and gain an HP of 8 whenever they kill deer. The tigers lose an HP of 0.021 at each step they do not eat a deer and die when they lose all the HP. In this game, the deer move randomly in the environment, and start with an HP of 5 and lose 1 HP upon attack. The tiger gets a reward of +1 for attacking a deer alongside another tiger. In this game, each tiger is assumed to get perfect information about the actions of other tigers at a distance of 4 units from itself and no information beyond this point. The view range of the tiger is 4 and the attack range is 1. The tigers also get a shaping reward of 0.2 for attacking a deer alone. All other rewards and conditions are the same as the Battle game. 

\subsection{Waterworld Domain}

The Waterworld domain was first introduced as a part of the Stanford Intelligent Systems Laboratory (SISL) environments by \citet{gupta2017cooperative}. We use the same domain adapted by the Petting Zoo environment \citep{terry2020pettingzoo}. This is a continuous action space environment, where a group of pursuer agents aim to capture food and avoid poison. These pursuers are the learning agents, while both food and poison move randomly in the environment. This is a cooperative environment where pursuer agents need to work together to capture food. At least two agents need to attack a food particle together to be able to capture it. The action is a two-element vector, where the first element corresponds to horizontal thrust and the second element corresponds to vertical thrust. The agents choose to use a thrust to make themselves move in a particular direction with a desired speed. The action values are in the range [-1, +1]. Our domain contains 25 pursuer agents, 25 food particles and 10 poison particles. The food is not destroyed but respawned upon capture. The agents get a reward of +10 upon capturing food and a penalty of -1 for encountering poison. If a single agent encounters food alone, it gets a shaping reward of +1. Each time an agent applies thrust, the agent obtains a penalty of thrust penalty $\times$ ||action||, where the value of thrust penalty is -0.5. Since this is a cooperative environment, each agent is allowed access to the global mean field action which is composed of actions of all agents in the environment, at each time step. This is done to simplify this complex domain. 

\subsection{Ride-Sharing Domain}\label{sec:ridesharingdetails}

In this domain, our problem formulation and environment are the same as that described in \citet{shah2020neural}. The demand distribution is obtained from the publicly available New York Yellow Taxi Dataset \citep{newyork2016}. The overall approach follows six steps. First, the user requests are obtained from the dataset. Second, sets of feasible trips are generated (by an oracle) using the approach in \citet{javier2017ondemand}. These feasible trips keep the action space from exploding. Third, the feasible actions are scored by the individual agents using their respective value functions. Fourth, a mapping of requests takes place by checking different constraints. Fifth, the final mapping is used to update the rewards for the individual agents. Sixth, the motion of vehicles is simulated until the next decision epoch.

We consider a maximum of 120 vehicles in our experiments. Since we are learning in a decentralized fashion, each vehicle maintains its own network and is computationally intensive. However, in practical applications, the training can be parallelized across agents and the computational demands will not be a limitation of our proposed setting. 

Our goal in this experiment is to implement the mean field algorithm on a real-world problem and compare the performance to other state-of-the-art approaches. The experimental setup is along the lines of \citet{javier2017ondemand, shah2020neural}, where we restrict ourselves to the street networks of Manhattan, where the vast majority of requests are contained. The New York Yellow Taxi dataset contains information about ride requests at different times of the day during a given week. Similar to \citet{shah2020neural}, we use the pickup and drop-off locations, pickup times and travel times from the dataset. The dataset has about 330,000 requests per day. In our experiments, the taxis are assigned a random location at the start and react to incoming ride requests. We train the networks using data pertaining to 8 weekdays and validate it on a single day as done in \citet{shah2020neural}. We assume all vehicles have similar capacities for simplicity, although our decentralized set-up is completely applicable to environments with very different types of vehicles as against prior work that relied on centralized training \citep{lowalekar2019zac, shah2020neural}. 

In the experiments, a single day of training corresponds to one episode. Each episode has 1440 steps, where each step (decision unit) is considered to span one minute. The initial location of vehicles at the beginning of an episode (single day of training) is random. The state, action space and reward function in our system is the same as that in \citet{shah2020neural}. The state of the system can be described as a tuple $(c_t, u_t)$, where $c_t$ is the location of each vehicle $j$ denoted by $c^j_t = (p^j, t, L^j)$ representing its trajectory. This captures the current location and a list of future locations that each vehicle visits. The user request $i$ at the time $t$ is represented as $u^i_t = (o^i, e^i)$ which corresponds to the origin and destination of the requests. The action for each agent is to provide a score for all the requests with the objective of assigning the user requests to itself. The user request should satisfy the constraints at the vehicle level (maximum allowed wait time and capacity limits) and the constraints at the system level (each request is only assigned to a single agent). Since these constraints are environmental, we continue to use a central agent that performs the constraints check before assigning requests to individual agents. Each agent gets a reward based on the proportion of requests in its feasible action set that it is able to satisfy, as done in \citet{shah2020neural}.

\section{Hyperparameters and Implementation Details}\label{sec:Hyperparameters}

The implementation of DMFG-QL, IL and MFQ almost use the same hyperparameters, with the learning rate set as $\alpha = 10^{-2}$. The temperature for the Boltzmann policy is set as 0.1.  Additionally, we also conduct epsilon greedy exploration which is decayed from 20\% to 1\% during the training process. The discount factor $\gamma$ is equal to 0.9. The replay buffer size is $2 \times 10^{5}$, and the agents use a mini-batch size of 64. The target network is updated at the end of every episode. 

Our implementations of DMFG-AC and MFAC almost use the same hyperparameters, where the learning rate of the critic is $10^{-2}$ and the learning rate of the actor is $10^{-4}$. The mean field network of the DMFG-AC uses a learning rate of $10^{-2}$. The discount factor is the same as the other three algorithms. In our implementation of MFAC, we do not use replay buffers unlike the implementation of \citet{pmlr-v80-yang18d}. We found this version of the actor-critic algorithm using the current updates (instead of delayed updates through the replay buffer) is more stable and performs better than the implementation of \citet{pmlr-v80-yang18d} in our experiments. Additionally, our implementations of both MFAC and DMFG-AC use complete decentralization during execution where the actors only need to use their local states (mean field does not need to be maintained anymore). Also since a separate network is being maintained for the stochastic policy (actor) we do not use the Boltzmann policy for the actor-critic based methods (MFAC and DMFG-AC).   

For the continuous action space Waterworld environment, every component of the obtained estimated mean field is normalized to be in the range $[-1,1]$ (the range of the action values) and we do not use a softmax for the output layer (since this a mixture of Dirac deltas as discussed in Section~\ref{sec:experiments}).  

Most hyperparameters are the same or closely match those used by prior works \citep{pmlr-v80-yang18d, Srirammtmfrl2020, guo2019learning, sriram2021partially} in many agent environments.

For the ride-sharing experiments, DMFG-QL uses the same network architecture as given in \citet{shah2020neural} for the NeurADP algorithm. The implementation of CO and NeurADP uses the implementation provided by \citet{shah2020neural} except that all agents are fully decentralized as mentioned before. Unlike the approach in \citet{shah2020neural}, each of our agents train their independent neural network using their local experiences. This network learns a suitable value function that can assign an appropriate score to each of the ride requests. This is done for both our implementations of NeurADP and DMFG-QL. For our baseline of CO \citep{javier2017ondemand}, we used the implementation from \citet{shah2020neural}, which used the immediate rewards as the score for the given requests along with a small bonus term pertaining to the time taken to process a request (faster processing of requests is encouraged). The mean field for the DMFG-QL implementation is obtained by processing a distribution of the ride requests across every node in the environment at each step. This mean field is made available to all agents during both training and testing. Also, we use a slightly different architecture for estimating the mean field in this domain (3 Relu layers of 50 nodes and an output softmax layer). 

The DDPG hyperparameters are based on \citet{lillicrap2015continuous} and the PPO hyperparameters are based on \citet{schulman2017proximal}. DDPG uses the learning rate of the actor as 0.001 and that of the critic as 0.002 and a discount factor of 0.9. We use the soft replacement strategy with learning rate 0.01. The batch size is 32. The PPO implementation uses the same batch size and discount factor. The actor learning rate is 0.0001 and critic learning rate is 0.0002. Independent PPO uses a single thread implementation, since the data correlations are already being broken by the non-stationary nature of the environment induced by the multiple agents. This is also computationally efficient. 

We use a set of 30 random seeds (1 -- 30) for all training experiments and a new set of 30 random seeds (31 -- 60) for all execution experiments.

\section{Complexity Analysis}\label{sec:complexityanalysisappendix}

A tabular version of our DMFG-QL algorithm is guaranteed to be linear in the total number of states, polynomial in the total number of actions, and constant in the number of agents. These guarantees are the same for both space complexity and time complexity. Comparatively, a tabular version of mean field $Q$-learning (MFQ) algorithm from \citet{pmlr-v80-yang18d} has a space complexity that is linear in the number of states, polynomial in the number of actions and exponential in the number of agents. This is due to Eq.~\ref{eq:updatemeana} requiring each agent to maintain $Q$-tables of all other agents. The time complexity is also the same as the space complexity in this case, since each entry in the table may need to be accessed in the worst case.

\section{More Related Work} \label{sec:morerealtedwork}

Mean field games have been used in the inverse reinforcement learning paradigm, where the objective is to determine a suitable reward function given expert demonstrations \citep{yang2017deep}. Model-based learning solutions have also been proposed for mean field games \citep{arman2013mean, yin2014learning}, though these methods suffer from restrictive assumptions on the model and the analysis does not generalize to other models. Methods such as \citet{cardaliaguet2017learning} perform a very computationally expensive computation using the full knowledge of the environment, which does not scale to large real-world environments. The work by \citet{fu2019actor} provides a mean field actor-critic algorithm along with theoretical analysis in the linear function approximation setting, unlike other works which only analyze the tabular setting \citep{guo2019learning, pmlr-v80-yang18d}. However, it is not clear if this algorithm is strong empirically since it does not contain empirical experiments that illustrate its performance. Further, this work considered the linear-quadratic setting that contains restrictions on the type of reward function.   \citet{mguni2018proceedings} explore the connection between MARL and mean field games in the model-free setting.

The mean field setting under the cooperative case can be handled using a mean field control model \citep{carmona2019linear}. In the linear-quadratic setting, \citet{carmona2019linear} prove that policy-gradient methods converge to local equilibrium. In further work, the same authors prove that $Q$-learning algorithms also converge \citep{carmona2019model}.

In the experiments, we consider a real-world application of our methods on the Ride-Pool Matching Problem (RMP) as originally defined in \citet{javier2017ondemand}. This is the problem considered by top ride-sharing platforms such as UberPool and Lyft-Line. The problem studies the efficiency of accepting ride requests by individual vehicles in such a way that the vehicles make more money (cater to more requests) per trip and the ride-sharing platform improves its ability to serve more orders. Previous approaches to solving the RMP problem have used standard optimization techniques \citep{ropke2009branch}. However, these methods are not scalable to large environments. One example from this class of methods is the zone path construction approach (ZAC) \citep{lowalekar2019zac}. The ZAC solves for an optimization objective where the environment is abstracted into zones and each zone path represents a set of trips. The available vehicles are assigned to suitable zone paths using the ZAC algorithm. Another proposed solution was to make greedy assignments \citep{javier2017ondemand}, which considers maximizing the immediate returns and not the long term discounted gains traditionally studied in RL. Prior work has also considered RL approaches for this problem \citep{Xu2018large, wang2018deep}. However, these works consider very restrictive settings, which do not model multiagent interactions between the different vehicles. The work by \citet{Li2019efficient} used a mean field approach for order dispatching, however, it assumes vehicles serve only one order at a time. The recent work by \citet{shah2020neural} introduced a very general formulation for the RMP where vehicles of arbitrary capacity are designed to serve batches of requests, with the whole solution scalable to thousands of vehicles and requests. They introduced a Deep $Q$-Network (DQN) \citep{mnih2015human} based RL approach (called neural approximate dynamic programming or NeurADP) that learns efficient assignment of ride requests. However, the proposed approach assumes a set of identical vehicles and uses centralized training, where all vehicles learn the same policy using centralized updates. This is not practical in real-world environments, where the individual vehicles are typically heterogeneous (many differences in vehicular capacity and preferences). We propose a mean field based decentralized solution to this problem, which is more practical than the approach by \citet{shah2020neural}.      

\section{Wall Clock Times}\label{sec:wallclocktimes}

The training for all the experiments on the simulated Petting Zoo domains was run on a 2 GPU virtual machine with 16 GB GPU memory per GPU. We use Nvidia Volta-100 (V100) GPUs for all these experiments. The CPUs use Skylake as the processor microarchitecture. We have a CPU memory of 178 GB. The Battle, Combined Arms and Gather experiments take an average of 5 days wall clock time to complete for all the considered algorithms. The Tiger-Deer experiments take an average of 4 days wall clock time to complete and the Waterworld experiments take an average of 2 days wall clock time to complete. 

The majority of our experiments on the RMP problem were run on a virtual machine with 4 Ampere-100 GPUs with a GPU memory of 40 GB each. The CPUs use Skylake as the processor microarchitecture. Each training takes an average of 5 days to complete execution.

\section{Statistical Significance Tests} 

We run a statistical significance test for all the main results in our paper. In the MAgent environments, for the training results we conduct an unpaired two-sided t-test at the last episode (2000) of training. For the execution results, we conduct a Fischer's exact test for the average performances. The tests are between the best performing algorithm (DMFG-QL or DMFG-AC) and the best performing baseline (IL, MFQ, MFAC). In the ride-sharing experiments we conduct an unpaired two-sided t-test for the average and standard deviation of performances in  Figure~\ref{fig:ridesharing}. The test is conducted between DMFG-QL and NeurADP (second best performing algorithm). We report the p-values for all the tests. As convention, we treat all p-values less than 0.05 as statistically significant outcomes. 

From the results in Table~\ref{table:statisticalsignificance}, we see that the better performance given by our algorithms (DMFG-QL or DMFG-AC) is statistically significant in all domains except the Tiger-Deer domain. As noted in Section~\ref{sec:experiments}, the MFQ and MFAC algorithms perform as well as the DMFG-QL and DMFG-AC algorithms in this cooperative domain. Though we observe that DMFG-QL gives the best overall average performance in both training and execution in the Tiger-Deer environment, the results are not statistically significant as noted from the p-values in Table~\ref{table:statisticalsignificance}. 

The statistical significance results for the ride-sharing experiment in Table~\ref{table:statisticalsignificanceridesharing} shows that our observations regarding the superior performance of DMFG-QL are statistically significant.

\begin{table}
\centering
\renewcommand{\arraystretch}{1.5}
\begin{tabular}{||p{0.20 \linewidth}  | p {0.18\linewidth}  | p {0.18\linewidth}  | p {0.21
\linewidth} ||} 
 \hline\hline
 Experimental Domain & Training (p-value) & Testing (p-value) & Statistically Significant \\ [0.5ex] 
 \hline\hline
 Battle  & p < 0.01 & p < 0.01 & Yes  \\ 
 \hline
  Combined Arms & p < 0.01 & p < 0.01 & Yes \\ 
 \hline
 Gather & p < 0.01 & p < 0.01 & Yes  \\ 
 \hline
  Tiger-Deer & p < 0.5 & p < 0.7 & No  \\
 \hline
   Waterworld & p < 0.01 & p < 0.01 & Yes \\
 \hline
 \hline
\end{tabular}
\caption{Statistical significance of our results in simulated experiments.}
\label{table:statisticalsignificance}
\end{table}

\begin{table}
\centering
\renewcommand{\arraystretch}{1.5}
\begin{tabular}{||p{0.11 \linewidth} | p{0.11 \linewidth} | p{0.11 \linewidth}  | p {0.18\linewidth}  | p {0.18\linewidth} ||} 
 \hline \hline
 \# of Vehicles ($N$) & Maxi-mum Pickup Delay ($\tau$) (sec) &  Capa-city ($c$) & p-value & Statistically Significant \\ [0.5ex] 
 \hline\hline
 80 & 580 &  10  & p < 0.01 & Yes   \\ 
 \hline
 100 & 580 & 10  &  p < 0.01 &  Yes \\ 
 \hline
 120 & 580 & 10  & p < 0.01 & Yes \\ 
 \hline \hline
   100 & 520 & 10 & p < 0.01 & Yes  \\
 \hline
   100 & 640 & 10 & p < 0.01 &  Yes \\
 \hline \hline
     100 & 580 & 8 & p < 0.01 &  Yes \\
 \hline
 100 & 580 & 12 & p < 0.02 &  Yes \\
 \hline
 \hline
\end{tabular}
\caption{Statistical significance of our results in the ride-sharing domain.}
\label{table:statisticalsignificanceridesharing}
\end{table}

\end{document}